\documentclass[11pt]{article}
\usepackage{fullpage}

\usepackage{libertine}

\usepackage{amsmath,amsthm,amssymb}
\usepackage{colortbl}
\usepackage{graphicx}
\usepackage{bm}

\usepackage{hyperref}
\usepackage[svgnames]{xcolor}
\usepackage[capitalise,nameinlink]{cleveref}
\hypersetup{colorlinks={true},linkcolor={DarkBlue},citecolor=[named]{DarkGreen}}
\usepackage{natbib}

\newtheorem{theorem}{Theorem}
\newtheorem{lemma}{Lemma}

\newtheorem{corollary}{Corollary}

\newtheorem{claim}{Claim}

\theoremstyle{definition}
\newtheorem{definition}{Definition}
\newtheorem{remark}{Remark}

\usepackage{multirow}
\usepackage{hhline}

\usepackage{xcolor}
\usepackage{hyperref}
\usepackage{comment,color}

\usepackage[ruled,vlined,linesnumbered]{algorithm2e}
\providecommand{\DontPrintSemicolon}{\dontprintsemicolon}
\makeatletter
\newcommand*{\algotitle}[2]{%
	\stepcounter{algocf}%
	\hypertarget{algocf.title.\theHalgocf}{}%
	\NR@gettitle{#1}%
	\label{#2}%
	\addtocounter{algocf}{-1}%
}
\makeatother

\usepackage{algpseudocode}
\usepackage{authblk}
\usepackage[capitalise]{cleveref}
\crefname{algorithm}{Mechanism}{mech}

\SetCommentSty{mycommfont}

\usepackage{multirow}

\allowdisplaybreaks

\makeatletter
\g@addto@macro\bfseries{\boldmath}
\makeatother

\newcount\Comments
\definecolor{darkgreen}{rgb}{0,0.7,0}
\newcommand{\kibitz}[2]{\ifnum\Comments=1\textcolor{#1}{#2}\fi}

\newcommand{\SW}{\text{SW}}
\newcommand{\myv}{\mathbf{v}}
\newcommand{\tildev}{\tilde{v}}
\newcommand{\tildevv}{\mathbf{\tilde{v}}}
\newcommand{\hatv}{\hat{v}}
\newcommand{\hatvv}{\mathbf{\hat{v}}}

\newcommand{\sucv}{\boldsymbol{\succ}_\myv}

\newcommand{\V}{\mathcal{V}}
\newcommand{\C}{\mathcal{C}}
\newcommand{\M}{\mathcal{M}}
\newcommand{\Q}{\mathcal{Q}}

\newcommand{\one}{\delta}

\DeclareMathOperator*{\argmax}{arg\,max}

\usepackage{tikz}

\protected\def\verythinspace{%
  \ifmmode
    \mskip0.5\thinmuskip
  \else
    \ifhmode
      \kern0.08334em
    \fi
  \fi
}

\let\OldSqrt\sqrt
\renewcommand{\sqrt}{\mskip-0.8\thinmuskip\OldSqrt}

\usepackage{etoolbox} %% <- for \pretocmd, \apptocmd and \patchcmd

\title{\bf Peeking Behind the Ordinal Curtain: \\ Improving  Distortion via Cardinal Queries\thanks{
A preliminary version of this paper appeared in {\em Proceedings of the 34th {AAAI} Conference on Artificial Intelligence ({AAAI})} ~\citep{conference}.
This work was partially supported by the European Research Council under the ERC projects ALGAME (grant no.~321171), ACCORD (grant no.~639945) and AMDROMA (grant no.~788893), by the NWO under the Gravitation project NETWORKS (grant no.~024.002.003) and the Veni project VI.Veni.192.153, by the Swiss National Science Foundation under contract number 200021\_165522, and by the Italian MIUR PRIN project ALGADIMAR.
}}
\author[1,2]{Georgios Amanatidis}
\author[3]{Georgios Birmpas}
\author[4]{Aris Filos-Ratsikas}
\author[5]{Alexandros A. Voudouris}

\affil[1]{Department of Mathematical Sciences, University of Essex, United Kingdom}
\affil[2]{ILLC, University of Amsterdam, The Netherlands}
\affil[3]{Department of Computer, Control and Management Engineering, Sapienza University of Rome, Italy}
\affil[4]{Department of Computer Science, University of Liverpool, United Kingdom}
\affil[5]{School of Computer Science and Electronic Engineering, University of Essex, United Kingdom}

\date{}

\begin{document}
\maketitle

\begin{abstract}
\noindent Aggregating the preferences of individuals into a collective decision is the core subject of study of social choice theory. In 2006, Procaccia and Rosenschein considered a utilitarian social choice setting, where the agents have explicit numerical values for the alternatives, yet they only report their linear orderings over them. To compare different aggregation mechanisms, Procaccia and Rosenschein introduced the notion of \emph{distortion}, which quantifies the inefficiency of using only ordinal information when trying to maximize the social welfare, i.e., the sum of the underlying values of the agents for the chosen outcome. Since then, this research area has flourished and bounds on the distortion have been obtained for a wide variety of fundamental scenarios. However, the vast majority of the existing literature is focused on the case where nothing is known beyond the ordinal preferences of the agents over the alternatives. In this paper, we take a more expressive approach, and consider mechanisms that are allowed to further ask \emph{a few cardinal queries} in order to gain partial access to the underlying values that the agents have for the alternatives. With this extra power, we design new \emph{deterministic} mechanisms that achieve significantly improved distortion bounds and, in many cases, outperform the best-known randomized ordinal mechanisms. We paint an almost complete picture of the number of queries required by deterministic mechanisms to achieve specific distortion bounds.
\end{abstract}

\section{Introduction}
Social choice theory~\citep{comsocbook2016} is concerned with aggregating the preferences of individuals into a joint decision. In an election, for instance, the winner should represent well (in some precise sense) the viewpoints of the voters. Similarly, the expenditure of public funds is typically geared towards projects that increase the well-being of society. Most traditional models assume that the preferences of individuals are expressed through \emph{ordinal preference rankings}, where each agent sorts all alternatives from the most to the least favorable according to her. Underlying these ordinal preferences, it is often assumed that there exists a \emph{cardinal} utility structure, which further specifies the intensity of the preferences \citep{vnm,BM:01,Barbera98}. That is, there exist numerical values that indicate how much an agent prefers an outcome to another. Given this cardinal utility structure, usually expressed via \emph{valuation functions}, one can define meaningful quantitative objectives, with the most prominent one being the maximization of the \emph{utilitarian} (or \emph{social}) \emph{welfare}, i.e., the sum of the values of the agents for the chosen outcome.

The main rationale justifying the dominance of ordinal preferences in the classical economics literature is that the task of asking individuals to express their preferences in terms of numerical values is arguably quite demanding for them. In contrast, performing simple comparisons between the different options is certainly more easily conceivable. To quantify how much the lack of cardinal information affects the maximization of quantitative objectives like the social welfare, \citet{procaccia2006distortion} defined the notion of \emph{distortion} for mechanisms as the worst-case ratio between the optimal social welfare (which would be achievable using cardinal information) and the social welfare of the outcome selected by the mechanism, which has access only to the preference rankings of the agents. Following their agenda, a plethora of subsequent works studied the distortion of mechanisms in several different settings, such as normalized valuation functions \citep{caragiannis2011embedding,boutilier2015optimal}, metric preferences \citep{anshelevich2018approximating,anshelevich2017randomized}, committee elections \citep{caragiannis2017subset}, and participatory budgeting \citep{benade2017preference}.

Somewhat surprisingly, the different variants of the distortion framework studied in this rich line of work differentiate between two extremes: we either have complete cardinal information or only ordinal information. Driven by the original motivation for using ordinal preferences, it seems quite meaningful to ask whether improved distortion guarantees can be obtained if one has access to \emph{limited} cardinal information, especially in settings for which the best-possible distortion bounds are already quite discouraging~\citep{boutilier2015optimal}. We formulate this idea via the use of \emph{cardinal queries}, which elicit cardinal information from the agents. These queries can be as simple as asking the value of an agent for a possible outcome, or even asking an agent \emph{whether an outcome is at least $x$ times better than some other outcome}, according to her underlying valuation function. Note that questions of the latter form are much less demanding than eliciting a complete cardinal utility structure, and thus are much more realistic as an elicitation device (see also the discussion below).

In this paper, we enhance the original distortion setting of \citet{procaccia2006distortion} and \citet{boutilier2015optimal} on single winner elections, by allowing the use of cardinal queries. In their setting, there are $n$ agents that have cardinal values over $m$ alternatives, and the goal is to elect a single alternative that (approximately) maximizes the social welfare, while having access \emph{only} to ordinal information. \citet{procaccia2006distortion} proved that no deterministic mechanism can achieve a distortion better than $\Omega(m)$ when agents have \emph{unit-sum} normalized valuation functions (i.e., the sum of the values of each agent for all possible alternatives is $1$), which was later on improved to $\Omega(m^2)$ by \citet{caragiannis2017subset}. Under the same assumption, \citet{boutilier2015optimal} proved that the distortion of any (possibly randomized) mechanism is between $\Omega(\sqrt{m}\verythinspace)$ and $O(\sqrt{m} \cdot \log^{*} m)$. Here we show how -- with only a limited number of cardinal queries -- deterministic mechanisms can significantly outperform any mechanism that has access only to ordinal information, even randomized ones.

\subsection{Our Contributions}
We initiate the study of trade-offs between the number of cardinal queries per agent that a mechanism uses and the distortion that it can achieve. In particular, we show results of the following type: 
\begin{quote}
\emph{The distortion $\mathcal{D}(\M)$ of a mechanism $\M$ that makes at most $\lambda$ queries per agent is $O(g(m, \lambda))$.}
\end{quote}
What our results suggest is that we can drastically reduce the distortion by exploiting only a small amount of cardinal information.

\subsubsection*{Query Model}
We consider two different types of cardinal queries, namely \emph{value queries} and \emph{comparison queries}. 
\begin{itemize}
	\item A \emph{value query} takes as input an agent $i$ and an alternative $j$, and returns the agent's value for that alternative. 
	\item A \emph{comparison query} takes as input an agent $i$, two alternatives $j,\ell$ and a real number $d$ and returns ``yes'' if the value of agent $i$ for alternative $j$ is at least $d$ times her value for alternative $\ell$, and ``no'' otherwise.
\end{itemize}
Note that value queries are qualitatively stronger than comparison queries, as they reveal much more detailed information. On the other hand, comparison queries are quite attractive as an elicitation device, since the cognitive complexity of the question that they pose is not much higher than that of forming a preference ranking. Additionally, comparison queries are conceptually similar to the idea of the original utility framework defined by \citet{vnm}. The idea there is that a cardinal scale for utility is possible because agents are capable of not only performing comparisons between alternatives, but also between lotteries over alternatives. For example, an agent $i$  should be able to tell whether she prefers alternative $a$ with certainty, or alternative $b$ with probability $1/2$ (where the remaining probability is assigned to a dummy alternative for whom all agents have value  $0$). Assuming \emph{risk-neutrality}, this is equivalent to asking the comparison query with parameters $(i,a,b,1/2)$.

\subsubsection*{Results and Techniques}
 We warm-up in \cref{sec:prefix} by using $\lambda$ simple prefix value queries per agent (i.e., ask her at the first $\lambda$ positions of her preference ranking). By selecting the alternative with the highest social welfare restricted to the query answers (the \emph{revealed welfare}), we obtain a linear improvement in the distortion, specifically $1+(m-1)/\lambda$. We show that this result is asymptotically optimal, among all mechanisms that use $\lambda$ prefix value queries per agent.

In \cref{sec:constant}, we devise a class of more sophisticated mechanisms that achieve much improved trade-offs between the distortion and the number of queries. In particular, our class contains
\begin{itemize}
	\item a mechanism that achieves \emph{constant} distortion using at most $O(\log^2{m})$ queries per agent, and
	\item a mechanism that achieves a distortion of $O(\sqrt{m}\verythinspace)$ using $O(\log{m})$ queries, matching the performance of the best possible randomized mechanism in the setting of \citep{boutilier2015optimal}, and outperforming all known randomized mechanisms for that setting.
\end{itemize} 
Our mechanisms are based on a binary search procedure, which for every agent finds the last alternative $\alpha$ in the agent's preference ranking such that the agent's value for $\alpha$ is at least $1/k$ times the value for her most-preferred alternative $\alpha^{*}$, for some chosen parameter $k$. Then, the mechanism \emph{simulates} the value of the agent for all alternatives that the agent ranks between $\alpha^{*}$ and $\alpha$ by her value for $\alpha$, and outputs the alternative that maximizes the \emph{simulated welfare}. By repeatedly applying this idea for appropriately chosen values of $k$, we explore the trade-offs between the distortion and the number of queries, when the latter range from $\log m$ to $\log^2 m$ per agent.
In \cref{sec:comparisons} we extend the above ideas to show that the mechanism which achieves a constant distortion using $O(\log^2 m)$ value queries, can actually be transformed into a mechanism which uses the same number of \emph{comparison queries}. In particular, we show how to approximate an agent's value for her most-preferred alternative using only $O(\log^2 m)$ comparison queries. 

In \cref{sec:lower-bounds} we present several lower bounds on the possible achievable trade-offs between the number of queries and distortion. These bounds follow by explicit instances where we carefully define a single ordinal preference profile as well as the cardinal information that may be revealed by the value queries of any mechanism. This information is defined in such a way so that, no matter how the mechanism makes its selection, it is always possible to create a superconstant gap between the optimal social welfare and the social welfare of the winning alternative. 

An overview of our main results can be found in \cref{table:results}.
We conclude the paper in \cref{sec:conclusions} with several interesting open problems, and a particular set of very challenging conjectures about the tight trade-offs between the number of queries and distortion.

\begin{remark}[Normalization assumptions]
We remark here that all of our upper bounds for value queries hold \emph{without any normalization assumption on the cardinal values}, contrary to the results of \citep{procaccia2006distortion} and almost all subsequent works in the related literature, which typically assume that values are normalized according to the unit-sum normalization. We do use the unit-sum normalization in \cref{sec:comparisons}, where we use comparison queries.\footnote{Actually, our results hold even if one uses other reasonable normalizations. For example, for the other common normalization assumption in the literature \citep{caragiannis2018truthful,Feige10,filos2014truthful}, the unit-range normalization, where the value of an agent for her most-preferred alternative is $1$ and all other values are in the interval $[0,1]$, the results of \cref{sec:constant} obviously extend verbatim to the case of comparison queries.} For the lower bounds, we prove bounds both for normalized and unrestricted values.
\end{remark}

\begin{remark}[Noisy queries]
Throughout this work we implicitly assume that agents can accurately answer all value queries. In fact, this is not necessary for any of our positive results! That is, we may assume that the answers to the queries are \emph{noisy}, e.g., because it requires extra effort for the agents to precisely determine these answers. As long as each inaccurate answer is at most a (multiplicative) constant factor away from the truth, all our upper bound proofs go through, at the expense of worse constants. Note that lower bounds are stronger when proven for exact queries, as is the case here.
\end{remark}

\renewcommand{\arraystretch}{1.5}
\setlength{\arrayrulewidth}{1pt}
\begin{table}[]
	\centering
	%\begin{small}
{\small 	\begin{tabular}{|c|c|c|}
		\hline
		\textbf{Number of queries}  & \textbf{Upper Bounds} & \textbf{Lower Bounds} \\ \hline
%		\cellcolor{lightgray!50!white}
		{$0$} {(ordinal, deterministic)} &  \cellcolor{lightgray!50!white}$O(m^2)$ \citep{caragiannis2011embedding} & 	\cellcolor{lightgray!50!white}$\Omega(m^2)$\ \  \citep{caragiannis2017subset} \\ \hline
%		\cellcolor{lightgray!50!white}
		{$0$} {(ordinal, randomized)} & 	\cellcolor{lightgray!50!white}$O(\sqrt{m} \log^{*}m)$ \ \ \citep{boutilier2015optimal} & 	\cellcolor{lightgray!50!white}$\Omega(\sqrt{m}\verythinspace)$\ \ \citep{boutilier2015optimal}  \\ \hline
		\multirow{2}{*}{{$1$ (value)}}& \multirow{2}{*}{$O(m)$\ \ [{$1$-PRV}, \cref{thm:prefix-upper}]}
		&$\Omega(m)$\ \  [\cref{thm:lower-one-unrestricted}] \\\hhline{~~-}
		& & 	\cellcolor{lightgray!50!white} $\Omega(\sqrt{m}\verythinspace)$\ \  [\cref{thm:lower-one-unitsum}]\\ \hline %      &        $O(m)$ (\textbf{$1$-RV}, \cref{thm:prefix-upper}) & $\Omega(m)$ (\cref{thm:lower-one-unrestricted})                                       \\ \hline
		%\textbf{$2$ (value)}               & $O(\sqrt{m}\verythinspace)$ (\textbf{$\sqrt{m}$-TRV}, \cref{thm:TRV})$^{*}$ & $\Omega(m^{1/6})$ (\cref{thm:lower-general-unrestricted})          \\ \hline
		{$\lambda \geq 2$  (value)} &      $O(m/\lambda)$\ \  [{$\lambda$-PRV}, \cref{thm:prefix-upper}] & $\Omega(m^{1/2(\lambda+1)})$\ \  [\cref{cor:lower-general-unrestricted}]                      \\ \hline
		
%		{$O\left(\frac{\log{m}}{\log\log{m}}\right)$ (value)} &      $O(\sqrt{m}\verythinspace)$\ \  [{$\sqrt{m}$-TRV}, \cref{thm:TRV}]$^\star$ & $\Omega(\log\log{m})$\ \  [\cref{cor:lower-general-unrestricted}] \\ \hline   
		
		{$O(\log m)$ (value)}       &    $O(\sqrt{m}\verythinspace)$\ \  [{$O(1)$-ARV}, \cref{cor:ARV}] &     
		\multirow{4}{*}{$\Omega(1)$}      \\ \cline{1-2}
		{$O(k \log m)$ (value)}       &    $O(m^{1/(k+1)}\verythinspace)$\ \  [{$k$-ARV}, \cref{thm:ARV-upper}] & \\	\cline{1-2}	 
		{$O(\log^2m)$ (value)}       &       $O(1)$\ \  [{$O(\log m)$-ARV}, \cref{cor:ARV}]  &            \\ \hhline{--~}
		%		\cellcolor{lightgray!50!white}
		{$O(\log^2m)$ (comparison)}  &      	\cellcolor{lightgray!50!white}$O(1)$\ \  [{$O(\log m)$-ARV}, \cref{cor:ARV-comparison}]      &   \\ \hline          
	\end{tabular}}
	%\end{small}
	\caption{A table showing the most important results in the paper. All our results are for \emph{deterministic} mechanisms. Results for unit-sum valuation functions are highlighted; everything else is for unrestricted valuation functions.}
	\label{table:results}
\end{table}

\subsection{Related Work}

The distortion framework was introduced by \citet{procaccia2006distortion}, and has been studied subsequently in a series of papers, most prominently by \citet{boutilier2015optimal}, who consider a general social choice setting, under the unit-sum normalization; this general model was also previously studied by \citet{caragiannis2011embedding} who considered different methods to translate the values of the agents for the alternatives into rankings (embeddings), and more recently by \citet{FMV19} who bounded the distortion of deterministic mechanisms in district-based elections. A related model is that of distortion of social choice functions in a metric space, which was initiated by \citet{anshelevich2018approximating}, and has since then been studied extensively~\citep{anshelevich2017randomized,goel2017metric,anshelevich2018ordinal,
borodin2019primarily,cheng2017people,cheng2018multiple,fain2018random,feldman2016facility,FRV20, ghodsi2019abstension,gkatzelis2020resolving,goel2018relating,gross2017agree,kempe2019analysis,
munagala2019improved, skowron2019approval}
In this setting, there is no normalization of values (or costs), but the valuation (or cost) functions are assumed to satisfy the triangle inequality. Similar distortion frameworks, in a metric space or under normalizations, have also been studied for other related problems, such as matching and clustering \citep{anshelevich2016blind,abramowitz2017utilitarians,anshelevich2017tradeoffs,anshelevich2018ordinal,Aris14}. It is worth noting that following the conference version of the present paper, recent works applied our research agenda and studied the tradeoffs between the number of queries and the distortion for the \emph{one-sided matching} problem and some of its generalizations~ \citep{amanatidis2020few,menon2020matching} .

Two related variants of the problem are $k$-winner elections, where $k$ alternatives are to be elected instead of one~\citep{caragiannis2017subset,benade2019low}, and participatory budgeting, where every alternative is associated with a cost, and one or more alternatives have to be elected in a manner that ensures that the total cost does not exceed a pre-specified budget constraint~\citep{lu2011budgeted}. \citet{benade2017preference} studied the $k$-winner participatory budgeting problem, but interestingly, they considered a more expressive model for the preferences of the agents, compared to simple preference rankings. In particular, they considered the knapsack votes model of \citep{goel2016knapsack}, rankings by value, rankings by value-for-money and threshold votes. While the first three are not very relevant for our purposes, the latter one can be thought of as a different type of (more expressive) query, in which a numerical value is specified, and every agent is asked to return the set of alternatives for which her value is above this threshold. Note that such threshold queries are in general incomparable to the value and comparison queries we consider in this work, as they elicit {\em aggregate} information with a single question.

Following \citet{benade2017preference}, \citet{bhaskar2018truthful} used a related model with thresholds drawn from $\mathcal{U}[0,1]$ to construct a randomized social choice function that approaches a distortion of $1$ with high probability as the number of agents approaches infinity. Their results are incomparable to ours for three main reasons: (a) threshold queries are stronger than value queries, as explained above, (b) their mechanisms are randomized, while we only consider deterministic mechanisms, and (c) their guarantees are obtained in the limit and in particular are not meaningful for cases where $m < n^3$; in constrast, we consider general values of $n$ and $m$. \citet{bhaskar2018truthful} also consider a different type of query, referred to as ``binary threshold'': an alternative and a threshold are selected and the voter is asked if her value for that alternative exceeds the threshold or not. Clearly, this is weaker than a value query. However, the algorithm of \cite{bhaskar2018truthful} requires $\Omega(m)$ binary threshold queries to achieve a constant approximation to the distortion; achieving constant distortion with $\Omega(m)$ value queries in our model is trivial.

Recently, \citet{mandalefficient,mandal2020optimal} studied a model conceptually related to ours, in which agents are asked to provide cardinal information, but there is a restriction on the number of bits to be communicated to the mechanism. Hence, they studied trade-offs between the number of transmitted bits and distortion. This is markedly quite different from what we do here, as a query in their setting has access to the (approximate) values of an agent for many alternatives simultaneously, and is therefore much too expressive when translated to our setting. On the other hand, the setting of \citet{mandalefficient,mandal2020optimal} does not assume ``free'' access to the ordinal preferences, which are also considered as part of the elicitation process. In particular, this implies that lower or upper bounds from our setting cannot be translated to theirs or vice-versa. 
We consider our work complementary to theirs, as they are mostly motivated by the computational limitations of elicitation (corresponding to a communication complexity approach), whereas we are motivated by the cognitive limitations of eliciting cardinal values, as often highlighted in the classical literature of social choice (corresponding to a query complexity approach). That being said, as we discussed earlier, our analysis is to some extent robust to noisy queries (which, for instance, could correspond to value declarations truncated to the $k$'th bit of their representation) and therefore the communication complexity of the problem could be studied on top of our query model.\footnote{Since the communication complexity is not the focus of this work, we have chosen to not pursue this further than the observation that constant inaccuracy leads to qualitatively similar results.}

Finally, at the same time and independently of the conference version of our work, \citet{abramowitz2019awareness} also introduced a setting in which the mechanism designer has access to some cardinal information on top of the ordinal preferences. This enables the design of improved mechanisms in terms of distortion. While the motivation of their paper is the same as ours, the approaches are inherently different. 
Besides the fact that \citet{abramowitz2019awareness} study a metric distortion setting, whereas we study a general setting with valuation functions that which are either unrestricted or normalized according to unit-sum, there is another fundamental distinction. The access to the cardinal information in \citep{abramowitz2019awareness} is not via queries. Instead, it is given explicitly as part of the input in terms of a threshold $\tau$, which allows the designer to know the number of agents for which the distance to an alternative $a$ is at most $1/\tau$ times their distance to another alternative $b$. 

\section{The model}\label{sec:model}

We consider a standard \emph{social choice setting}, in which there is a set $A$ of $m$ alternatives and a set $N$ of $n$ agents. Our goal is to elect a \emph{single} alternative based on the preferences of the agents, which are expressed through {\em valuation functions} $v_i: A \rightarrow \mathbb{R}_{\geq 0}$ that map alternatives to non-negative real numbers. 
For notational convenience, we use $v_{ij}$ instead of $v_i(j)$ to denote the \emph{cardinal value} of agent $i$ for alternative $j$, and refer to the matrix $\myv = (v_{ij})_{i \in N, j \in A}$ as a \emph{valuation profile}. By $\mathbf{V}$ we denote the set of all possible valuation profiles.
Clearly, the valuation function $v_i$ also defines a {\em preference ranking} for agent $i$, i.e., a linear ordering $\succ_{i}$ of $A$ such that $j \succ_i j'$ if $v_{ij} \geq v_{ij'}$; we assume that ties are broken according to a deterministic tie-breaking rule, e.g., according to a fixed global ordering of the alternatives.\footnote{It would be equivalent to allow ties at this point, get pre-linear orderings instead, and leave the tie-breaking to the mechanisms when necessary.}
We refer to $\sucv  = (\succ_{1}, \ldots, \succ_{n})$ as an (ordinal) \emph{preference profile}.

In this work, we consider the following two families of valuation functions:
\begin{itemize}
\item \emph{Unrestricted valuation functions}, which may take any non-negative real values.

\item \emph{Unit-sum valuation functions}, which are such that $\sum_{j \in A} v_{ij} =1$ for every agent $i \in N$.
\end{itemize}

The \emph{social welfare} of alternative $j \in A$ with  respect to $\myv$ is the total value of the agents for $j$: $\SW(j \,|\, \myv)=\sum_{i \in N} v_{ij}$. Our goal is to output one of the alternatives who maximize the social welfare, i.e., an alternative in $\argmax_{j \in A} \SW(j \,|\,\myv)$. This is clearly a trivial task if one has full access to the valuation profile. However, we assume \emph{limited access} to these cardinal values. In particular, we assume that we only have access to the preference profile $\myv_{\succ}$ and can also learn cardinal information by asking queries. We consider two types of queries: \emph{value queries} that reveal the value of an agent for a given alternative, and \emph{comparison queries} that reveal whether the value of an agent for an alternative is a multiplicative factor larger than her value for some other alternative.

\begin{definition} Given a preference profile, a query about the underlying cardinal values is called 
	\begin{itemize}
		\item A \emph{value query}, if it takes as input an agent $i$ and an alternative $j$ and returns the agent's value $v_{ij}$ for that alternative. This is implemented via the function $\V : N \times A \rightarrow \mathbb{R}_{\geq 0}$.
		We say that \emph{agent $i$ is queried at position $k$}, if alternative $j$ is ranked $k$-th in $\succ_{i}$ and we make the query $\V(i,j)$.
		
		\item A \emph{comparison query}, if it takes as input an agent $i$, two alternatives $j$, $\ell$ and a real number $d$, and returns \textsc{yes} if $v_{ij}\geq d \cdot  v_{i\ell}$, and \textsc{no} otherwise. This is implemented via the function $\C: N \times A \times A \times \mathbb{R}_{\geq 0} \rightarrow \{\textsc{yes},\textsc{no}\}$.
	\end{itemize} 
\end{definition}

Clearly, value queries reveal more information than comparison queries. Note that the information obtained by a comparison query can be obtained by at most two value queries. On the other hand, however, without any cardinal information or any normalization assumption, it is impossible to even approximate the information obtained by a value query using only comparison queries. In this sense, value queries are considerably stronger than comparison queries.

\begin{definition}
A \emph{mechanism} $\M=(\Q,f)$ with access to a (value or comparison) oracle takes as input a preference profile $\sucv$ and returns an alternative. In particular, it consists of the following two parts:

\begin{itemize}
	\item An algorithm $\Q$ that takes as input the preference profile $\sucv$, adaptively makes queries to the oracle, and returns the set of answers to these queries.
	\item A mapping $f$ that takes as input the preference profile $\sucv$ and the set $\Q(\sucv)$ of answers to the queries above, and outputs a single alternative $j \in A$. Such a mapping is called a \emph{social choice function}.
\end{itemize} 
\end{definition}
 
By the description of $\Q$ above, it is clear that the mechanism is free to choose the positions at which each agent will be queried, and those can depend not only on $\sucv$, but on the answers to the queries already asked as well. 
The performance of a mechanism is measured by its \emph{distortion}, which is the worst-case ratio---over all possible instances---between the optimal social welfare and the social welfare of the alternative chosen by the mechanism. 

\begin{definition}
The \emph{distortion} of a mechanism $\M$ is 
\[  \mathcal{D}(\M) = \sup_{(N,A,\myv)} \frac{\max_{j \in A}\SW(j\,|\,\myv)}{\SW(\M(\sucv)\,|\,\myv)} \,,  \]
where $\SW(j\,|\,\myv)$ is the social welfare of alternative $j$ given a particular valuation profile,
and $\M(\myv_{\succ})$ is the output of the mechanism on input $\sucv$.
\end{definition}

Throughout our proofs, it will be useful to partition the quantity $\SW(j\,|\,\myv)$, into two separate quantities depending on the cardinal information we  obtain from the queries. This is particularly relevant when we deal with value queries, but even for comparison queries we use a similar decomposition in \cref{sec:comparisons}.

\begin{definition}
The \emph{revealed welfare $\SW_r(j\,|\,\myv)$ of $j$} is the contribution to $\SW(j\,|\,\myv)$ of agents that have been queried for alternative $j$ via value queries, i.e.,  $\SW_r(j\,|\,\myv) = \sum_{i \in N: \V(i, j) \in \Q(\sucv)} v_{ij}$. 
The remaining quantity $\SW(j\,|\,\myv) - \SW_r(j\,|\,\myv)$ is called the \emph{concealed welfare $\SW_c(j\,|\,\myv)$ of $j$}.
\end{definition}

\section{Warm-Up: Mechanisms Using Fixed-Position Value Queries}\label{sec:prefix}
Before we dive into our more technical results, we first warm-up by discussing probably the most obvious class of mechanisms, which query every agent at the first $\lambda \geq 1$ positions; we refer to such queries as {\em prefix}. A particular member of this class is the mechanism that uses the {\em Range Voting} (RV) social choice function to decide the outcome. Formally, RV takes as input the whole valuation profile $\myv$ and elects an alternative $x$ with maximum social welfare: $x \in \argmax_{j \in A} \SW(j\,|\,\myv)$. In our case, since $\myv$ is not fully known, we deploy RV only on the revealed valuation profile, where any unknown value is assumed to be zero. 

To be more specific, let $T_k(j)$ be the set of agents that rank alternative $j \in A$ at position $k \in [m]$. Our mechanism first queries every agent at each of the first $\lambda$ positions of her preference ranking. Then, it elects the alternative $y$ that maximizes the revealed welfare: $y \in \argmax_{j \in A} \SW_r(j\,|\,\myv)$. We refer to this mechanism as {\em $\lambda$-Prefix Range Voting} (\nameref{alg:PRV}).
We remark that a very similar mechanism was independently proposed by \citet{mandalefficient}; the analyses of the two mechanisms follow along the same lines. 

\begin{algorithm}[ht]
	\DontPrintSemicolon 
	\NoCaptionOfAlgo
	\algotitle{$\lambda$-PRV}{alg:PRV.title}

	\For{$j \in A$ }{
		$\SW_r(j \,|\, \myv)= 0$  %\tcc*{{\footnotesize initialize the revealed welfares}} 
	}
	\For{$i \in N$}{
	    \For{$k \in [\lambda]$}{
	    	Ask the query $\V(i, j_i(k))$ to learn $v_{i, j_i(k)}$, where $j_i(k)$ is the $k$-th favorite alternative of  $i$. \;
	    	$\SW_r(j_i(k) \,|\, \myv)= \SW_r(j_i(k) \,|\, \myv) + v_{i, j_i(k)}$  %\tcc*{{\footnotesize update the revealed welfares}} 
	    }	
	}
	
	Let $y \in \argmax_{j\in A} \SW_r(j \,|\, \myv)$ be an alternative achieving the best revealed welfare. \;
	\Return $y$

	\caption{Mechanism $\lambda$-PRV$(\sucv)$} 
	\label{alg:PRV} 
\end{algorithm}\medskip 

\begin{theorem}\label{thm:prefix-upper}
The distortion of $\lambda$-\rm{PRV} is $\mathcal{D}(\lambda\text{-\rm{PRV}})\leq 1 + \frac{m-1}{\lambda}$, even for unrestricted valuation functions.
\end{theorem}

\begin{proof}
Consider some instance with valuation profile $\myv$. Let $x$ be an alternative that maximizes the social welfare according to $\myv$, and let $y$ be the alternative that is elected by $\lambda$-PRV. 
Recall that here the revealed welfare of any alternative $j \in A$ is $\SW_r(j\,|\,\myv) = \sum_{k=1}^\lambda \sum_{i \in T_k(j)} v_{ij}$. 
Since $\SW(y \,|\, \myv) \geq \SW_r(y \,|\, \myv)$, it suffices to show that $\SW(x\,|\,\myv) \leq \big(1 + \frac{m-1}{\lambda}\big)\, \SW_r(y \,|\, \myv)$. To this end, we will bound the revealed and the concealed welfare of $x$ separately. 

Since $y$ is an alternative that maximizes the revealed welfare, we have that $\SW_r(y \,|\, \myv) \geq \SW_r(j \,|\, \myv)$ for every $j \in A$, and therefore 
\begin{align}\label{eq:prefix-opt-revealed}
\SW_r(x \,|\, \myv) \leq \SW_r(y \,|\, \myv)\,.
\end{align}
Now, consider the agents in $\bigcup_{k = \lambda+1}^m T_k(x)$. They are not queried about their value for $x$, and therefore contribute to the concealed welfare of $x$. For every such agent $i$ there exist $\lambda$ different alternatives $j_i(1), \ldots, \allowbreak j_i(\lambda)$ that $i$ ranks above $x$, and for whom she has value $v_{i,j_i(1)}, \ldots , v_{i,j_i(\lambda)} \geq v_{ix}$.\footnote{When the subscripts have subscripts themselves, we follow the common practice of separating them with commas.} Consequently, we have that
\begin{align} \nonumber
\SW_c(x \,|\, \myv) 
&= \sum_{k=\lambda+1}^m \sum_{i \in T_k(x)} v_{ix} 
\leq \sum_{k=\lambda+1}^m \sum_{i \in T_k(x)} \frac{v_{i,j_i(1)} + \ldots + v_{i,j_i(\lambda)}}{\lambda} \\
\label{eq:prefix-opt-concealed}
&= \frac{1}{\lambda} \sum_{j \in A \setminus\{ x \}} \SW_r(j \,|\, \myv) \leq \frac{1}{\lambda} \sum_{j \in A \setminus\{x\}} \SW_r(y \,|\, \myv) 
= \frac{m-1}{\lambda}\, \SW_r(y \,|\, \myv) \,.
\end{align}
The statement now follows by \eqref{eq:prefix-opt-revealed} and \eqref{eq:prefix-opt-concealed}.
\end{proof}

Clearly, the distortion guarantee of $\lambda$-PRV improves linearly in the number of queries $\lambda$. Nevertheless, it is interesting to see for which values of $\lambda$ the mechanism achieves distortion $O(\sqrt{m}\verythinspace)$ and $O(1)$. These are given by the following statement. 

\begin{corollary}\label{cor:prefix}
The distortion of $\lambda$-\rm{PRV} is 
\begin{align*}
\mathcal{D}(\lambda\textrm{-PRV}) =
\begin{cases}
O(\sqrt{m}\verythinspace), & \text{ for } \lambda = \Theta(\sqrt{m}\verythinspace) \\
O(1), & \text{ for } \lambda = \Theta(m)
\end{cases}
\end{align*}
\end{corollary}

Next, we show that, in terms of distortion, $\lambda$-PRV is the best possible mechanism among those that make at most $\lambda$ prefix value queries. 

\begin{theorem}\label{thm:prefix-lower}
Any mechanism that makes $\lambda$ prefix value queries per agent has distortion $\Omega(m/\lambda)$, even for unit-sum valuation functions. 
\end{theorem}

\begin{proof}
Consider an instance with $n$ agents and $m=n$ alternatives $a_1, \ldots, a_m$. Let $\lambda \leq m/2$. We define the following ordinal profile:
\begin{itemize}
\item The $\lambda$ favorite alternatives of agent $i$ are $a_i, a_{i+1}, \ldots, a_{i+\lambda-1}$
%$a_i$, $a_{(i+1)\!\!\! \mod{m}}$, \ldots, $a_{(i+\lambda-1)\!\!\! \mod{m}}$ 
in decreasing order, 
where all the indices are considered modulo $m$. 
Hence, all alternatives appear exactly once at each of the first $\lambda$ positions.

\item Alternatives $x=a_1$ and $y=a_{\lambda+1}$ appear $m/2$ times each at position $(\lambda+1)$, in the $m-\lambda \geq m/2$ agent rankings in which they do not appear at the first $\lambda$ positions. Observe that, by definition, $x$ and $y$ do not appear together at the first $\lambda$ positions in any  preference ranking, and there are multiple ways to decide in which rankings each of them appears at position $(\lambda+1)$; any such construction works for our purposes. 

\item For every agent, the remaining alternatives are arbitrarily ordered at positions $(\lambda+2)$ up to $m$.
\end{itemize}
See \cref{tab:prefix-lower-example} for a specific example of the ordinal profile. 

\begin{table}[th]
\centering
\begin{tabular}{c c c c c c c}
\noalign{\hrule height 0.5pt}\hline
agent &  \multicolumn{6}{c}{ranking} \\ 
\noalign{\hrule height 0.5pt}\hline
1 & $a_1$ & $a_2$ & $a_3$ & $a_4$ & $a_5$ & $a_6$ \\
2 & $a_2$ & $a_3$ & $a_1$ & $a_4$ & $a_5$ & $a_6$ \\
3 & $a_3$ & $a_4$ & $a_1$ & $a_2$ & $a_5$ & $a_6$ \\
4 & $a_4$ & $a_5$ & $a_3$ & $a_1$ & $a_2$ & $a_6$ \\ 
5 & $a_5$ & $a_6$ & $a_1$ & $a_3$ & $a_2$ & $a_4$ \\
6 & $a_6$ & $a_1$ & $a_3$ & $a_2$ & $a_4$ & $a_5$ \\
\noalign{\hrule height 0.5pt}\hline
\end{tabular} 
\caption{An example of the ordinal profile used in the proof of \cref{thm:prefix-lower} with $m=n=6$ and $\lambda=2$, where $x=a_1$, $y=a_3$, $T_{\lambda+1}(x) = \{2,3,5\}$ and $T_{\lambda+1}(y)=\{1,4,6\}$.}
\label{tab:prefix-lower-example}
\end{table}

The valuation profile $\myv$ is such that each agent has value $\frac{1}{\lambda+1}$ for her first $\lambda$ favorite alternatives.  It is without loss of generality to assume that any mechanism that knows the ordinal information of this instance and also makes $\lambda$ prefix value queries, must elect either $x$ or $y$. To see this, first notice that given the revealed cardinal information, the revealed welfare of all alternatives is the same. Further, given the particular preference profile, it is easy to always complete the valuation profile  $\myv$ in a way that guarantees that no alternative has more concealed welfare than $x$ and $y$; indeed, the two possible valuation profiles we consider have this property.  

So, assume that the mechanism selects alternative $y$ (the case of $x$ being completely symmetric). Now the remaining values of the agents are such that the $m/2$ agents in $T_{\lambda+1}(x)$ have value $\frac{1}{\lambda+1}$ for $x$ and $0$ for the remaining alternatives, while the $m/2$ agents in $T_{\lambda+1}(y)$ have value $\frac{1}{(m-\lambda)(\lambda+1)}$ for all alternatives at positions $\lambda+1$ up to $m$.

Given this valuation profile $\myv$, the social welfare of the winner $y$ is
\begin{align*}
\SW(y\,|\,\myv) = \frac{\lambda}{\lambda+1} + \frac{ \frac{m}{2}}{(m-\lambda)(\lambda+1)} \leq 1 \,.
\end{align*}  
In contrast, the social welfare of the optimal alternative $x$ is
\begin{align*}
\SW(x\,|\,\myv) = \frac{\lambda + \frac{m}{2}}{\lambda+1} + \frac{ \frac{m}{2} }{(m-\lambda)(\lambda+1)} \geq \frac{m}{2(\lambda+1)} \,.
\end{align*}
Therefore, the distortion of any mechanism is at least $\frac{m}{2(\lambda+1)}$.
\end{proof}

We now turn our attention to a slightly more general class of mechanisms which query all agents at the same fixed positions, and show that $\lambda$-PRV remains best possible among the mechanisms of this class for unrestricted valuation functions. In \cref{sec:lower-bounds} we further show that $1$-PRV is best possible among all mechanisms that make one query per agent for unrestricted valuation functions.

\begin{theorem}\label{thm:fixed-lower}
For unrestricted valuation functions, any mechanism that makes $\lambda$ fixed-position value queries per agent has distortion $\Omega(m/\lambda)$. 
\end{theorem}

\begin{proof}
Let $\lambda \leq m/2$. Consider any mechanism of this class, and let $\ell$ be the first position at which it does not query the agents. Observe that if $\ell > \lambda$, then the mechanism only makes prefix value queries. In this case, the bound follows by \cref{thm:prefix-lower}, which holds for unit-sum valuation functions, and thus for unrestricted ones as well. So, we may assume that $\ell\in[\lambda]$.

Now, we consider an instance with $n=m$ that is very similar to the one presented in the proof of \cref{thm:prefix-lower}. Essentially, we substitute $(\lambda+1)$ with $\ell$, and we have that all alternatives appear exactly once at each of the first $\ell-1$ positions, while two alternatives $x$ and $y$ appear $m/2$ times each at position $\ell$. The remaining alternatives for every agent are arbitrarily ordered at position $\ell+1$ up to $m$.

The valuation profile $\myv$ is such that each agent has value $1$ for her first $\ell-1$ favorite alternatives, and value $0$ for the alternatives at positions $\ell+1$ up to $m$. Observe that the revealed welfare of all alternatives is exactly equal to $\ell-1$.
Given the revealed cardinal information and the particular ordinal profile, we can argue exactly like we did in the proof of \cref{thm:prefix-lower} about fact that it is without loss of generality to assume that the mechanism elects either $x$ or $y$. 
So, assume that the mechanism selects alternative $y$; the case of $x$ is symmetric. The remaining values of the agents are such that the $m/2$ agents in $T_{\ell}(x)$ have value $1$ for $x$, while the $m/2$ agents in $T_\ell(y)$ have value $0$ for $y$.

Given this valuation profile $\myv$, the social welfare of the winner $y$ is
$\SW(y\,|\,\myv) = \ell-1 \leq \lambda-1$, while the social welfare of the optimal alternative $x$ is
$\SW(x\,|\,\myv) = \ell-1 + \frac{m}{2} \geq \frac{m}{2}$. 
Therefore, the distortion of the mechanism is $\Omega(m/\lambda)$.
\end{proof}

\section{Improving Distortion via Simulated Valuation Functions}\label{sec:constant}
Our goal in this section is to further explore the additional power that cardinal queries provide, and focus on the design of mechanisms with improved distortion guarantees.
Mechanism $\lambda$-PRV is a good first step in this direction, but it needs to make a  large number of queries per agent in order to do so; in particular, by \cref{cor:prefix}, it achieves distortion $O(\sqrt{m}\verythinspace)$ for $\lambda=\Theta(\sqrt{m}\verythinspace)$ and constant distortion for $\lambda=\Theta(m)$. 
Therefore, it is natural to ask whether it is possible to design mechanisms that achieve similar distortion bounds, but require much less queries per agent. 
We  answer this question positively.

For any $k \in [m]$, we define a mechanism which we call {\em $k$-Acceptable Range Voting} (\nameref{alg:ARV}).
Let $\lambda_0, \lambda_1, \ldots, \lambda_k$ be $k+1$ thresholds such that $\lambda_{\ell} = m^{\frac{\ell}{k+1}}$ for $\ell \in \{0, 1, \ldots, k\}$. For every agent $i \in N$, we first query her value $v_i^*$ for her favorite alternative $j_i(1)$. Then, using binary search we compute the maximal {\em $\lambda_{\ell}$-acceptable set} $S_{i,\ell} = \{j \in A: v_{ij} \geq v_i^*/\lambda_\ell \}$ for every $\ell \in [k]$; see the procedure {BSearch} in the pseudocode describing the mechanism. Also, for each agent $i$, we define $S_{i,0}=\{j_i(1)\}$ to contain only $i$'s favorite alternative. The $\lambda_{\ell}$-acceptable set of an agent consists of the alternatives that this agent finds at most $\lambda_\ell$ times worse than her favorite alternative. We continue by constructing a new approximate valuation profile $\tildevv$, where the values of every agent $i$ are
\begin{itemize}
\item $\tildev_i^* = v_i^*$;
\item $\tildev_{ij} = v_i^*/ \lambda_\ell$ for every  $j \in S_{i,\ell} \setminus S_{i,\ell-1}$ with $\ell \in [k]$;
\item $\tildev_{ij} = 0$ for every $j \in A \setminus S_{i,k}$.
\end{itemize}
We finally elect the alternative $z \in A$ that maximizes the social welfare according to the approximate valuation profile: $z \in \arg\max_{j \in A} \sum_{i \in N} \tildev_{ij}$.

\begin{algorithm}[hp]
	\DontPrintSemicolon 
	\NoCaptionOfAlgo
	\algotitle{$k$-ARV}{alg:ARV.title}

        \SetKwFunction{proc}{\normalfont{\textsf{BSearch}}}
	\vspace{0.2cm}

	\For{$i \in N$ }{
		 $v_i^* = \V(i, j_i(1))$, where $j_i(p)$ is the alternative that agent $i$ ranks at position $p$. \label{line:vi*}\;
		 $\tilde{v}_{i,j_i(1)}=v_i^*$ \;
         $S_{i,0} = \{ j_i(1) \}$ \;
   		\For{$\ell \in \{1, 2, \dots, k\}$ }{
        	 $\lambda_{\ell} = m^\frac{\ell}{k+1}$ \;
        	 $p^* = $\ \ \proc{$1,m,\lambda_{\ell},v_i^*$}\\
			 $S_{i,\ell} = \{j \in A:  j \succ_i j_i(p^*) \}$  \tcc*{define the $\lambda_\ell$-acceptable set of agent $i$} 
            \For {$j \in S_{i,\ell} \setminus S_{i,\ell-1}$}{  
            		  $\tilde{v}_{ij} = v_i^*/\lambda_\ell$ \tcc*{define the approximate valuation profile}
            }
    	}	
    	\For {$j \in A \setminus S_{i,k}$}{
             $\tilde{v}_{ij} = 0$ \;
        }
    } 
    \For{$j \in A$}{ 
    	 $\SW_s(j \,|\, \tilde\myv) = 0$  \tcc*{compute the simulated welfare of alternative $j$} 
    	\For{$j \in A$}{
			 $\SW_s(j \,|\, \tilde\myv) = \SW_s(j \,|\, \tilde\myv) + \tilde{v}_{ij}$ 
		}
	}	
	Let $z \in \argmax_{j\in A} \SW_s(j \,|\, \tilde\myv)$ be an alternative achieving the best simulated welfare.\;
	\Return $z$

	\vspace{10pt}
	
	\SetKwProg{myproc}{Procedure}{}{}
	\myproc{\proc{$\alpha$, $\beta$, $\lambda$, $v$}}{
		\If{$\alpha=\beta$}
		{
			\Return $\alpha$
		}
		Let $u = \V(i,j_i(\frac{\alpha+\beta}{2}))$\\
		\If{$u \geq v/\lambda$} 
		{\vspace{0.2cm}
			\proc{$\frac{\alpha+\beta}{2},\beta,\lambda,v$}
		}
		\Else{
			\proc{$\alpha,\frac{\alpha+\beta}{2},\lambda,v$}
		}

	}
	
	\caption{Mechanism $k$-ARV$(\sucv)$} \label{alg:ARV} 
\end{algorithm}

Now, we proceed by proving an upper bound on the distortion achieved by $k$-ARV as a function of $k$. 

\begin{theorem}\label{thm:ARV-upper}
The mechanism $k$-\rm{ARV} makes $O(k\log m)$ value queries per agent, and has distortion $\mathcal{D}(k\text{-ARV})=O(\!\sqrt[\uproot{2}k+1]{m}\verythinspace)$.
\end{theorem}

\begin{proof}
Consider any instance with valuation profile $\myv$.
Since mechanism $k$-ARV executes a binary search in order to compute the $\lambda_{\ell}$-acceptable sets for each $\ell \in [k]$, it requires a total of $O(k \log{m})$ value queries per agent. The rest of the proof is dedicated in bounding the distortion of $k$-ARV.
First, we fix our notation: 
\begin{itemize}
\item $z$ is the alternative elected by $k$-ARV.
\item $y$ is a welfare-maximizing alternative for the valuation profile $\hatvv$, which is such that the value of agent $i \in N$ for alternative $j \in A$ is
\begin{align*}
\hatv_{ij} =
\begin{cases}
0, & \text{ if } j \in A \setminus S_{i,k} \\
v_{ij}, & \text{otherwise}.
\end{cases}
\end{align*}
That is, $y \in \arg\max_{j \in A} \sum_{i \in N} \hatv_{ij}$.

\item $x$ is the welfare-maximizing alternative for the true  profile $\myv$. That is, $x \in \arg\max_{j \in A} \sum_{i \in N} v_{ij}$.
\end{itemize}
Also, for a valuation profile $\mathbf{u}$, let $N_j(\mathbf{u}) = \{i\in N: u_{ij}>0\}$ be the set of agents with strictly positive value for alternative $j\in A$ with respect to  $\mathbf{u}$. We use the following easy fact about welfare-maximizing alternatives. 

\begin{lemma}\label{lem:N_j}
If $j^* \in \arg\max_{j \in A} \sum_{i \in N} v_{ij}$, then $j^* \in \arg\max_{j \in A} \sum_{i \in N_j(\mathbf{v})} v_{ij}$.
\end{lemma}	
	
To prove the statement, we will bound the social welfare of $x$ in terms of the social welfare of $z$ for the true valuation profile $\myv$. In particular, we will show that
	\begin{align}\label{eq:sw-x-z}
	\SW(x\,|\,\mathbf{v}) \leq \left( \lambda_1 + \frac{m}{\lambda_{k}} \right) \SW(z\,|\,\mathbf{v}) \,.
	\end{align}
Then, the approximation ratio of $k$-ARV will be 
	\begin{align*}
	\frac{\SW(x\,|\,\mathbf{v})}{\SW(z\,|\,\mathbf{v})} \leq \lambda_1 + \frac{m}{\lambda_{k}} = 2\cdot m^\frac{1}{k+1} = O(\!\sqrt[\uproot{2}k+1]{m}\verythinspace)\,.
	\end{align*}
	
We partition the social welfare of $x$ into the following two quantities: the contribution of the agents $i$ that place $x$ in the $\lambda_k$-acceptable set $S_{i,k}$, and the contribution of the remaining agents that have small value for $x$. By definition, we have that $i \in N_x(\hatvv)$ for any agent $i$ such that $x \in S_{i,k}$, and therefore
\begin{align*}
\SW(x\,|\,\myv) = \sum_{i \in N_x(\hatvv)} v_{ix} + \sum_{i \notin N_x(\hatvv)} v_{ix}
\end{align*}
We first consider the term $\sum_{i \in N_x(\hatvv)} v_{ix}$, and have that
	\begin{align}\label{eq:x-maximal}
	\sum_{i \in N_x(\hatvv)} v_{ix} 
	\leq \sum_{i \in N_y(\hatvv)} v_{iy} 
	\leq \lambda_1  \sum_{i \in N_y(\hatvv)} \tildev_{iy} 
	\leq \lambda_1  \sum_{i \in N_z(\tildevv)} \tildev_{iz} 
	\leq \lambda_1  \sum_{i \in N_z(\tildevv)} v_{iz} 
	\leq \lambda_1 \, \SW(z\,|\,\myv) \,,
	\end{align}
where 
\begin{itemize}
\item the first inequality follows by the fact that $\SW(j \,|\, \hatvv) = \sum_{i \in N_j(\hatvv)}  \hatv_{ij} = \sum_{i \in N_j(\hatvv)} v_{ij}$, the definition of $y$ as the alternative that maximizes the social welfare for the valuation profile $\hatvv$, and \cref{lem:N_j};

\item for the second inequality it suffices to notice that for any $i \in N_y(\hatvv)$ there exists an $\ell \in [k]$ such that $y \in S_{i,\ell} \setminus S_{i,\ell-1}$, and thus $v_{ij} \le \frac{v_i^*}{\lambda_{\ell-1}} = \lambda_1  \frac{v_i^*}{\lambda_{\ell}} = \lambda_1  \tildev_{ij}$;

\item the third inequality follows by the definition of $z$ as the alternative that maximizes the social welfare for the valuation profile $\tildevv$ 
and \cref{lem:N_j};

\item the fourth inequality follows by the fact that $v_{ij} \geq \tildev_{ij}$, for every $i \in N$ and $j \in A$;

\item the last inequality follows trivially from the fact that $N_z(\tildevv)\subseteq N$.
\end{itemize}

\noindent 
Next, we consider the term $\sum_{i \notin N_x(\hatvv)} v_{ix}$. By the definition of $N_x(\hatvv)$, for every $i \not\in N_x(\hatvv)$ it holds that $x \not\in S_{i,k}$, and hence $v_{ix} < v_i^*/\lambda_k$. Using this, we obtain
\begin{align}\label{eq:x-top-bound}
	\sum_{i \notin N_x(\hatvv)} v_{ix} 
	< \sum_{i \notin N_x(\hatvv)} \frac{v_i^*}{\lambda_k} 
	= \frac{1}{\lambda_k}\sum_{i \notin N_x(\hatvv)} v_i^* 
	\le \frac{1}{\lambda_k}\sum_{i \in N\setminus T_1(x)} v_i^* 
	= \frac{1}{\lambda_k} \sum_{j \in A\setminus\{x\}} \sum_{i \in T_1(j)}v_{ij} \,,
\end{align}
where recall that $T_1(j)$ is the set of agents whose favorite alternative is $j$, and for whom $v_{i}^*= \tildev_{i}^* =\tildev_{ij} = v_{ij}$. 
Since $z$ is the alternative that maximizes the quantity $\sum_{i \in N} \tildev_{ij}$, for every $j$ we have that
\begin{align*}
	\sum_{i \in N} \tildev_{iz} 
	\geq \sum_{i \in N} \tildev_{ij} 
	= \sum_{i \in T_1(j)} v_{ij} + \sum_{i \in N\setminus T_1(j)} \tildev_{ij} 
	\geq \sum_{i \in T_1(j)} v_{ij} \,.
\end{align*}
Combining the above inequality together with the fact that $v_{iz} \geq \tildev_{iz}$ for every agent $i \in N$, we have that
\begin{align*}
	\sum_{i \in N} v_{iz} \geq \sum_{i \in T_1(j)} v_{ij} \,.
\end{align*}
Using this last inequality, \eqref{eq:x-top-bound} becomes 
\begin{align}\label{eq:x-non-maximal}
	\sum_{i \notin N_x(\hatvv)} v_{ix} 
	\le \frac{1}{\lambda_k} \sum_{j \in A\setminus\{x\}} \sum_{i \in T_1(j)}v_{ij}
	\leq \frac{1}{\lambda_k} \sum_{j \in A\setminus\{x\}} \sum_{i \in N} v_{iz} 
	= \frac{m-1}{\lambda_k} \, \SW(z\,|\, \mathbf{v}).
\end{align}
Finally, the desired inequality \eqref{eq:sw-x-z} follows by combining inequalities \eqref{eq:x-maximal} and \eqref{eq:x-non-maximal}. 
\end{proof}

The next statement follows by appropriately setting the value of the parameter $k$ in \cref{thm:ARV-upper}, and shows how mechanism $k$-ARV improves upon the distortion guarantees of $\lambda$-PRV using way less value queries per agent.

\begin{corollary}\label{cor:ARV}
We have that
\begin{itemize}
\item $1$-\rm{ARV} achieves distortion $O(\sqrt{m}\verythinspace)$ using $O(\log{m})$ values queries per agent;
\item $\log{m}$-\rm{ARV} achieves distortion $O(1)$ using $O(\log^2{m})$ value queries per agent. 
\end{itemize}
\end{corollary}

\noindent We conclude this section by showing that the analysis of $k$-ARV is tight. 

\begin{theorem}\label{thm:ARV-lower}
The distortion of $k$-\rm{ARV} is $\Omega(\!\sqrt[\uproot{2}k+1]{m}\verythinspace)$.
\end{theorem}

\begin{proof}
Recall that $\lambda_1 = m^{\frac{1}{k+1}} = \lambda$ and consider the following instance with $m$ alternatives $A=\{a_1, ..., a_m\}$ and $n=m-2$ agents. 
To simplify our discussion, let $z=a_{m-1}$ and $x=a_m$. 
The valuation profile $\myv$ is such that the values of agent $i$ are
\begin{itemize}
\item $v_{i,a_i} = v_{ix} = \frac{\lambda}{2\lambda+1}$, 
\item $v_{iz} = \frac{1}{2\lambda+1}$, and
\item $v_{i,a_j} = 0$ for $j \in [m] \setminus \{i,m-1,m\}$. 
\end{itemize}
In the ordinal profile $\sucv$ which is given as input to the mechanism, we assume without loss of generality that agent $i$ ranks alternative $a_i$ ahead of $x$.
 
Since $\frac{1}{2\lambda+1} = \frac{1}{\lambda}\cdot \frac{\lambda}{2\lambda+1}$, 
$k$-ARV defines only one acceptable set per agent using $\lambda$.
In particular, the algorithm sets $S_{i,1} = \{x,z\}$ for every agent $i \in [m-2]$.
Then, the approximate valuation profile $\tildevv$ is such that the values of agent $i$ are
\begin{itemize}
\item $\tildev_{i,a_i} = \frac{\lambda}{2\lambda+1}$, 
\item $\tildev_{ix} = \tildev_{iz} =\frac{1}{2\lambda+1}$, and
\item $\tildev_{i,a_j} = 0$ for $j \in [m] \setminus \{i,m-1,m\}$
\end{itemize}
For the approximate valuation profile $\tildevv$, the social welfare of both alternatives $x$ and $z$ is 
$$\SW(x \,|\, \tildevv) = \SW(z \,|\, \tildevv) = \frac{m-2}{2\lambda+1},$$ 
while any other alternative $j \in A \setminus \{x,z\}$ has social welfare 
$$\SW(j \,|\, \tildevv) = \frac{\lambda}{2\lambda+1}.$$
Hence, $k$-ARV might select alternative $z$ as the winner instead of $x$, and the distortion is then 
\begin{align*}
\frac{(m-2)\frac{\lambda}{2\lambda+1}}{(m-2)\frac{1}{2\lambda+1}} = \lambda = \!\sqrt[\uproot{2}k+1]{m},
\end{align*}
as desired.
\end{proof}

\subsection{Implementing $k$-ARV with Comparison Queries}\label{sec:comparisons}
A crucial observation is that mechanism $k$-ARV can actually be implemented using just \emph{one} value query.  We can ask the value of each agent for her favorite alternative, and then ask $O(k\log{m})$ comparison queries that guide the binary search in computing the maximal acceptable sets.  Hence, $\log m$-ARV achieves constant distortion using only one value query and $O(\log^2{m})$ comparison queries. Therefore, it is natural to ask whether we can avoid this single value query entirely, and rely only on comparison queries instead. Surprisingly, for unit-sum valuation functions, we show that this is indeed possible at no extra cost! More precisely, we show that we can approximate the value that an agent has for her favorite alternative within a factor of $1\pm\varepsilon$, using $O(\log^2{m})$ comparison queries. 
Note that this is the only time that we assume the unit-sum normalization for any of our upper bounds.

For the sake of readability, we focus on a single agent and write $u_j$ for her value for the alternative that she ranks at position $j \in [m]$. We take the same approach as in the proof of \cref{thm:ARV-upper} in order to build an approximate valuation profile. Since everything in this profile is expressed in terms of the largest value $u_1$, we utilize the unit-sum assumption to approximately solve for $u_1$.

\begin{theorem}\label{thm:comparison}
For any constant $\varepsilon\in [1/m,1]$, it is possible to compute some $u^*$ such that $(1-\varepsilon)\, u^* \leq u_1 \leq (1+\varepsilon)\, u^*$, using $O\big(\frac{\log^2{m}}{\varepsilon}\big)$ comparison queries per agent.
\end{theorem} 

\begin{proof}
Let $\kappa = \lceil \log_{1+\varepsilon}{m^2} \rceil = \Theta\big(\frac{\log{m}}{\log{(1+\varepsilon)}}\big) = \Theta\big(\frac{\log{m}}{\varepsilon}\big)$. 
We define $\kappa$ thresholds $\lambda_{\ell} = (1+\varepsilon)^{\ell}$ for ${\ell} \in [\kappa]$;
observe that $\lambda_i = \lambda_{i-1} \cdot \lambda_1$.
For each ${\ell} \in [\kappa]$, we perform a binary search using $\Theta(\log{m})$ comparison queries to find the maximum integer $\xi_{\ell}$ such that $u_{\xi_{\ell}} \geq \frac{u_1}{\lambda_{\ell}}$; we also set $\xi_0=1$ and $\xi_{\kappa+1}=m$. 
Hence, we have that 
\begin{align}\label{eq:utility-bounds}
u_j \in 
\begin{cases}
\left[ \frac{u_1}{\lambda_{\ell}}, \frac{u_1}{\lambda_{{\ell}-1}} \right), & \text{for each } j \in (\xi_{{\ell}-1},\xi_{\ell}], {\ell} \in [\kappa] \\\\
\left[0, \frac{u_1}{\lambda_\kappa} \right), & \text{for each } j \in (\xi_\kappa,\xi_{\kappa+1}] \,. 
\end{cases}
\end{align}
To simplify the notation, let $g_i = \xi_i - \xi_{i-1}$ for $i \in [\kappa+1]$. Note that $g_i \leq m$.  

By the unit-sum normalization we have
\begin{align*}
1 = \sum_{j=1}^m u_j = \sum_{{\ell}=1}^{\kappa+1} \sum_{j \in (\xi_{{\ell}-1},\xi_{\ell}]} u_j \,.
\end{align*}
Using \eqref{eq:utility-bounds}, we can now upper- and lower-bound the above expression. We start with the upper bound:
\[
1 = \sum_{{\ell}=1}^{\kappa+1} \sum_{j \in (\xi_{{\ell}-1},\xi_{\ell}]} u_j
\leq \sum_{{\ell}=1}^{\kappa+1} \sum_{j \in (\xi_{{\ell}-1},\xi_{\ell}]} \frac{u_1}{\lambda_{{\ell}-1}} \\
= \sum_{{\ell}=1}^{\kappa+1} u_1 \frac{g_{\ell}}{\lambda_{{\ell}-1}} = u_1 \sum_{{\ell}=1}^{\kappa} \frac{g_{\ell}}{\lambda_{{\ell}-1}} + u_1 \frac{g_{\kappa+1}}{\lambda_\kappa} \,.
\]
By the definition of $\lambda_\kappa$ we have that $\frac{g_{\kappa+1}}{\lambda_\kappa} \leq \frac{1}{m} \leq \varepsilon$ and hence,
\begin{align*}
u_1 \geq (1 - \varepsilon) \cdot \left( \sum_{{\ell}=1}^{\kappa} \frac{g_{\ell}}{\lambda_{{\ell}-1}} \right)^{-1}.
\end{align*}
Similarly, by using the lower bounds in \eqref{eq:utility-bounds}, we have that
\[
1 = \sum_{{\ell}=1}^{\kappa+1} \sum_{j \in (\xi_{{\ell}-1},\xi_{\ell}]} u_j
\geq \sum_{{\ell}=1}^{\kappa} \sum_{j \in (\xi_{{\ell}-1},\xi_{\ell}]} \frac{u_1}{\lambda_{{\ell}}} \\
= \sum_{{\ell}=1}^{\kappa} u_1 \frac{g_{\ell}}{\lambda_{{\ell}}} = u_1 \sum_{{\ell}=1}^{\kappa} \frac{g_{\ell}}{\lambda_{{\ell}-1}} \cdot \frac{1}{\lambda_1}
=   \frac{u_1}{1+\varepsilon}\sum_{{\ell}=1}^{\kappa} \frac{g_{\ell}}{\lambda_{{\ell}-1}} \,.
\]
or, equivalently,
\begin{align*}
u_1 \leq (1 + \varepsilon) \cdot \left( \sum_{{\ell}=1}^{\kappa} \frac{g_{\ell}}{\lambda_{{\ell}-1}} \right)^{-1}.
\end{align*}
Hence, the theorem follows by setting $u^* = \left( \sum_{{\ell}=1}^{\kappa} \frac{g_{\ell}}{\lambda_{{\ell}-1}} \right)^{-1}$. Indeed, to compute $u^*$ only the integers $\xi_{\ell}$, ${\ell} \in [\kappa]$ are used. Those, in turn, are computed via $\Theta(\log{m})$ binary searches. Hence, we only need $O\big(\frac{\log^2{m}}{\varepsilon}\big)$ comparison queries.
\end{proof}

By inspecting the proof of \cref{thm:ARV-upper}, it is easy to see that knowing the approximate valuation profile $\tildevv$ exactly or perturbed within a multiplicative constant factor, makes no difference asymptotically. Therefore, we augment $k$-ARV with a pre-processing step where each maximum value $v_i^*$ is approximated according to \cref{thm:comparison} above, and these approximations are used in line \ref{line:vi*} of the mechanism.  
For $k = \log m$, this new mechanism, which we call \emph{modified} $(\log m)$-ARV, achieves the same distortion guarantee and asks the same number of queries (asymptotically) as $(\log m)$-ARV.

\begin{corollary}\label{cor:ARV-comparison}
Modified $(\log{m})$-\rm{ARV} achieves distortion $O(1)$ using $O(\log^2{m})$ comparison queries per agent. 
\end{corollary}

% Note that the factors of the comparison queries used by the modified $k$-ARV mechanism are in general fractional. One could argue that it is easier for the agents to answer comparison queries with integral factors. It is not hard\cref{cor:ARV-comparison} stil

\section{Lower Bounds} \label{sec:lower-bounds}
We now present general lower bounds on the distortion which depend on the number of value queries the mechanisms are allowed to ask per agent, but are unconditional on how and where they decide to ask these queries. In particular, we show that the distortion of any mechanism that makes one value query per agent is $\Omega(m)$ when the agents have unrestricted valuation functions, and $\Omega(\sqrt{m}\verythinspace)$ when the agents have unit-sum valuation functions. Moreover, for mechanisms that are allowed to make $\lambda \geq 1$ queries per agent, we show a weaker lower bound of $\Omega\left(\frac{1}{\lambda+1} \cdot m^{\frac{1}{2(\lambda+1)}}\right)$ for unrestricted valuation functions. This shows that in order to achieve constant distortion, we need to necessarily make $\omega\left(\frac{\log{m}}{\log\log{m}}\right)$ queries per agent. Closing the gap between this lower bound and the upper bound of $O(\log^2{m})$ queries that the mechanism $O(\log{m})$-ARV from \cref{sec:constant} requires in order to achieve constant distortion is one of the most interesting problems that our work leaves open; see the discussion in \cref{sec:conclusions}.

Before we proceed with the presentation of the results of this section, let us give a very brief roadmap of the proofs. The high-level idea is similar to those used in the lower bound proofs presented in previous sections (for example, see \cref{thm:prefix-lower} and \cref{thm:fixed-lower} in \cref{sec:prefix}), but the particular constructions and arguments exploited in the following proofs are more delicate; this is a consequence of the fact that we aim to lower-bound the distortion of any mechanism. To this end, assuming an arbitrary mechanism (that is allowed to make a specific number of queries per agent), we first define a single ordinal preference profile which is given as input to the mechanism, and also carefully define the cardinal information that could be revealed by any query of the mechanism. This cardinal information is such that it is always possible to define the unknown part of the valuation profile in a way that leads to a large enough gap between the optimal social welfare and the social welfare of the alternative selected by the mechanism. Since we do not know how the mechanism makes its selection, we need to take into account every possible scenario, and therefore define many different valuation profiles that can be used in different cases.\\

\noindent \textbf{Remark.} To simplify our discussion when we deal with unrestricted valuation functions in this section, we assume that the values are normalized and lie in the interval $[0,1]$. This is without loss of generality since we make no other assumption about the way the mechanisms behave, other than that they are allowed to ask a particular number of queries.

\subsection{One-Query Mechanisms with Unrestricted Valuations}
We start by showing that, for unrestricted valuations, any mechanism that makes one value query per agent has linear distortion. This also shows that the mechanism $1$-PRV from \cref{sec:prefix} is the best possible mechanism among such mechanisms. 

\begin{theorem}\label{thm:lower-one-unrestricted}
For unrestricted valuation functions, the distortion of any mechanism that uses one value query per agent is $\Omega(m)$.
\end{theorem}

\begin{proof}
Let $\M$ be an arbitrary mechanism that makes one value query per agent, and consider an instance with $m \geq 4$ alternatives and $n=m-2$ agents, where $m$ is an even number. 
We denote the set of alternatives as $A=\{a_1, ..., a_{m-2},x,y\}$. 
Using the notation $[z, w]$ to denote the fact that alternatives $z$ and $w$ are ordered arbitrarily in the ranking of an agent, we define the ordinal profile as follows:
\begin{itemize}
\item The ranking of agent $i \leq \frac{n}{2}$ is $a_i \succ_i x \succ_i y \succ_i [ a_1, ..., a_{i-1}, a_{i+1}, ..., a_{m-2} ]$;

\item The ranking of agent $i > \frac{n}{2}$ is $a_i \succ_i y \succ_i x \succ_i [ a_1, ..., a_{i-1}, a_{i+1}, ..., a_{m-2} ]$.
\end{itemize} 
Depending on the positions at which $\M$ queries, we reveal the following cardinal information:
\begin{itemize}
\item For every query at a first position we reveal a value of $m^{-1}$;
\item For every query at a second or third position we reveal a value of $m^{-2}$;
\item For any other position we reveal a value of $0$. 
\end{itemize}

We claim that $\M$ must query all agents at the first position, as otherwise its distortion is $\Omega(m)$. Assume otherwise that $\M$ does not query agent $1$ her value for alternative $a_1$; this is without loss of generality due to symmetry. We now define two valuation profiles $\myv_1$ and $\myv_2$, which are both consistent to the ordinal profile and the revealed information, but differ on the value that agent $1$ has for alternative $a_1$. In particular:
\begin{itemize}
\item In both $\myv_1$ and $\myv_2$, every agent $i \geq 2$ has value $m^{-1}$ for alternative $a_i$, $m^{-2}$ for alternatives $x$ and $y$, and $0$ for everyone else;

\item In both $\myv_1$ and $\myv_2$, agent $1$ has value $m^{-2}$ for alternatives $x$ and $y$, and $0$ for every alternative $a_i$ for $i \geq 2$. The value of agent $1$ for alternative $a_1$ is $m^{-2}$ in $\myv_1$, and $1$ in $\myv_2$. 
\end{itemize}
These two profiles are utilized in the following way: If $\M$ selects $a_1$, then the valuation profile is set to be $\myv_1$, while if $\M$ selects some other alternative, then the valuation profile is set to be $\myv_2$.
Now, observe that
\begin{align*}
\SW(a_i \,|\, \myv_1) = \SW(a_i \,|\, \myv_2) = m^{-1} \text{ \ \ \ \ for every $i \geq 2$},
\end{align*} 
and 
\begin{align*}
\SW(x \,|\, \myv_1) = \SW(x \,|\, \myv_2) = \SW(y \,|\, \myv_1) = \SW(y \,|\, \myv_2) = (m-2)\cdot m^{-2} \leq m^{-1}.
\end{align*}
If $\M$ selects $a_1$, the social welfare of $a_1$ is $\SW(a_1 \,|\, \myv_1) = m^{-2}$ and therefore any alternative $a_i$ for $i \geq 2$ is optimal, yielding distortion equal to $m$. Similarly, when $\M$ selects some alternative different than $a_1$, then $a_1$ is optimal with social welfare $\SW(a_1 \,|\, \myv_2)=1$, yielding distortion at least $m$. 

Hence, $\M$ must query all agents at the first position in order to learn a value of $m^{-1}$ for every alternative $a_i$, $i \in [n]$. We now define three valuation profiles $\myv_3$, $\myv_4$ and $\myv_5$, which are consistent to the ordinal profile and this revealed information, but differ on the values that the agents have for alternatives $x$ and $y$; in particular, $\myv_4$ and $\myv_5$ are symmetric.
\begin{itemize}
\item In all three profiles, every agent $i \in [n]$ has value $m^{-1}$ for alternative $a_i$, and $0$ for any alternative $a_j$ such that $j \neq i$;

\item In $\myv_3$, all agents have value $m^{-1}$ for alternatives $x$ and $y$;

\item In $\myv_4$, all agents have value $m^{-2}$ for alternative $y$, every agent $i > n/2$ (who ranks $x$ after $y$) has value $m^{-2}$ for $x$, and every agent $i \leq n/2$ (who ranks $x$ before $y$) has value $m^{-1}$ for $x$.

\item In $\myv_5$, all agents have value $m^{-2}$ for alternative $x$, every agent $i \leq n/2$ (who ranks $y$ after $x$) has value $m^{-2}$ for $y$, and every agent $i > n/2$ (who ranks $y$ before $x$) has value $m^{-1}$ for $y$.
\end{itemize} 
If $\M$ selects some alternative $a_i$ for $i \in [n]$, then the valuation profile is set to be $\myv_3$, while if $\M$ selects alternative $y$ or $x$, then the valuation profile is set to be $\myv_4$ or $\myv_5$, respectively.
Given this, observe that if $\M$ decides to select alternative $a_i$ for some $i \in [n]$, then since 
\begin{align*}
\SW(a_i \,|\, \myv_3) = \SW(a_i \,|\, \myv_4) = \SW(a_i \,|\, \myv_5) = m^{-1} \text{ \ \ \ \ for every $i \in [n]$},
\end{align*} 
and 
\begin{align*}
\SW(x \,|\, \myv_3) = \SW(y \,|\, \myv_3) = (m-2)\cdot m^{-1} = 1- 2m^{-1},
\end{align*}
the distortion is at least $m-2$. 
Similarly, if $\M$ decides to select alternative $y$, then since 
\begin{align*}
\SW(y \,|\, \myv_4) = (m-2)m^{-2} \leq m^{-1}
\end{align*}
and 
\begin{align*}
\SW(x \,|\, \myv_4) = \left( \frac{m}{2}-1 \right) m^{-1} + \left( \frac{m}{2}-1 \right) m^{-2} = \frac{1}{2}\bigg( 1 - m^{-1}- 2m^{-2} \bigg),
\end{align*}
the distortion is at least $\frac{1}{2}(m - 1 - 2m^{-1}) \geq \frac{m}{4}$ for any $m \geq 4$; the case where $\M$ selects $x$ is symmetric and follows by $\myv_5$. In any case, $\M$ has distortion $\Omega(m)$ and the theorem follows.
\end{proof}

\subsection{General Mechanisms with Unrestricted Valuations}
\label{sec:lower-bounds_lambda}
We will now focus on mechanisms that make a number $\lambda \geq 1$ of queries per agent, and will show a weaker lower bound on the distortion which depends on $\lambda$.

\begin{theorem}\label{thm:lower-general-unrestricted}
For unrestricted valuation functions, the distortion of any mechanism that uses $\lambda \geq 1$ value queries per agent is $\Omega\left( \frac{1}{\lambda+1} \cdot m^{\frac{1}{2(\lambda+1)}}\right)$. 
\end{theorem}

\begin{proof}
Our instance consists of $m \geq \lambda$ alternatives and $n=m$ agents.
We partition the set $A$ of alternatives into the following $(\lambda+2)$ sets:
\begin{itemize}
\item $A_j$ with $|A_j|=m^{1-\frac{j}{\lambda+1}}$, for every $j \in [\lambda]$;
\item $X$ with $|X|=2$;
\item $Y$ with $|Y|=m - 2 - \sum_{j=1}^\lambda m^{1-\frac{j}{\lambda+1}}$.
\end{itemize}
We will now define the ordinal profile. 
Let $X=\{x_1,x_2\}$ and denote by $[z,w]$ the fact that alternatives $z$ and $w$ are ordered arbitrarily in the ranking of an agent. For every agent $i$ there exists an alternative $a_{ij} \in A_j$ for each $j \in [\lambda]$ such that:
\begin{itemize}
\item The ranking of agent $i \leq \frac{m}{2}$ is $a_{i1} \succ_i ... \succ_i a_{i\lambda} \succ_i x_1 \succ_i x_2 \succ_i [Y] \succ_i [\cup_{j \in [\lambda]} A_j \setminus \{a_{ij}\}]$;
\item The ranking of agent $i > \frac{m}{2}$ is $a_{i1} \succ_i ... \succ_i a_{i\lambda} \succ_i x_2 \succ_i x_1 \succ_i [Y] \succ_i [\cup_{j \in [\lambda]} A_j \setminus \{a_{ij}\}]$;
\end{itemize}
In words, every agent $i$ ranks some alternative $a_{ij} \in A_j$ at position $j \in [\lambda]$, followed by the two alternatives of $X=\{x_1,x_2\}$ at positions $(\lambda+1)$ and $(\lambda+2)$, followed by all alternatives of $Y$ (in an arbitrary order), followed by the alternatives of $\cup_{j\in [\lambda]} A_j \setminus \{a_{ij}\}$ (in an arbitrary order). Observe that the alternatives of $Y$ are all dominated by the alternatives of $X$ in the sense that both $x_1$ and $x_2$ have at least as much social welfare as any alternative of $Y$. 
The choices as to how the alternatives of $\cup_{j \in [\lambda]} A_j$ are distributed in the rankings of the agents  are such that:  
\begin{itemize}
\item Each alternative of $A_j$ appears $m^{\frac{j}{\lambda+1}}$ times at position $j \in [\lambda]$;
\item For any $j\in [\lambda-1]$ and pair of agents $i,i'$ such that $a_{ij}=a_{i'j}$, it holds that $a_{i,j+1}=a_{i',j+1}$.
\end{itemize}
Hence, the agents with the same favorite alternative have exactly the same ranking. To simplify our discussion in what follows, we refer to the alternatives in $Y$ and $\cup_{j \in [\lambda]} A_j \setminus \{a_{ij}\}$ as the {\em tail} alternatives of agent $i$. Let $T_j(z)$ be the set of the $m^{\frac{j}{\lambda+1}}$ agents that rank alternative $z \in A_j$ at position $j \in [\lambda]$. \cref{fig:instance-unrestricted} depicts the ordinal profile of our instance for $\lambda=2$. 

\begin{figure}[ph!]
\centering
\includegraphics[scale=1]{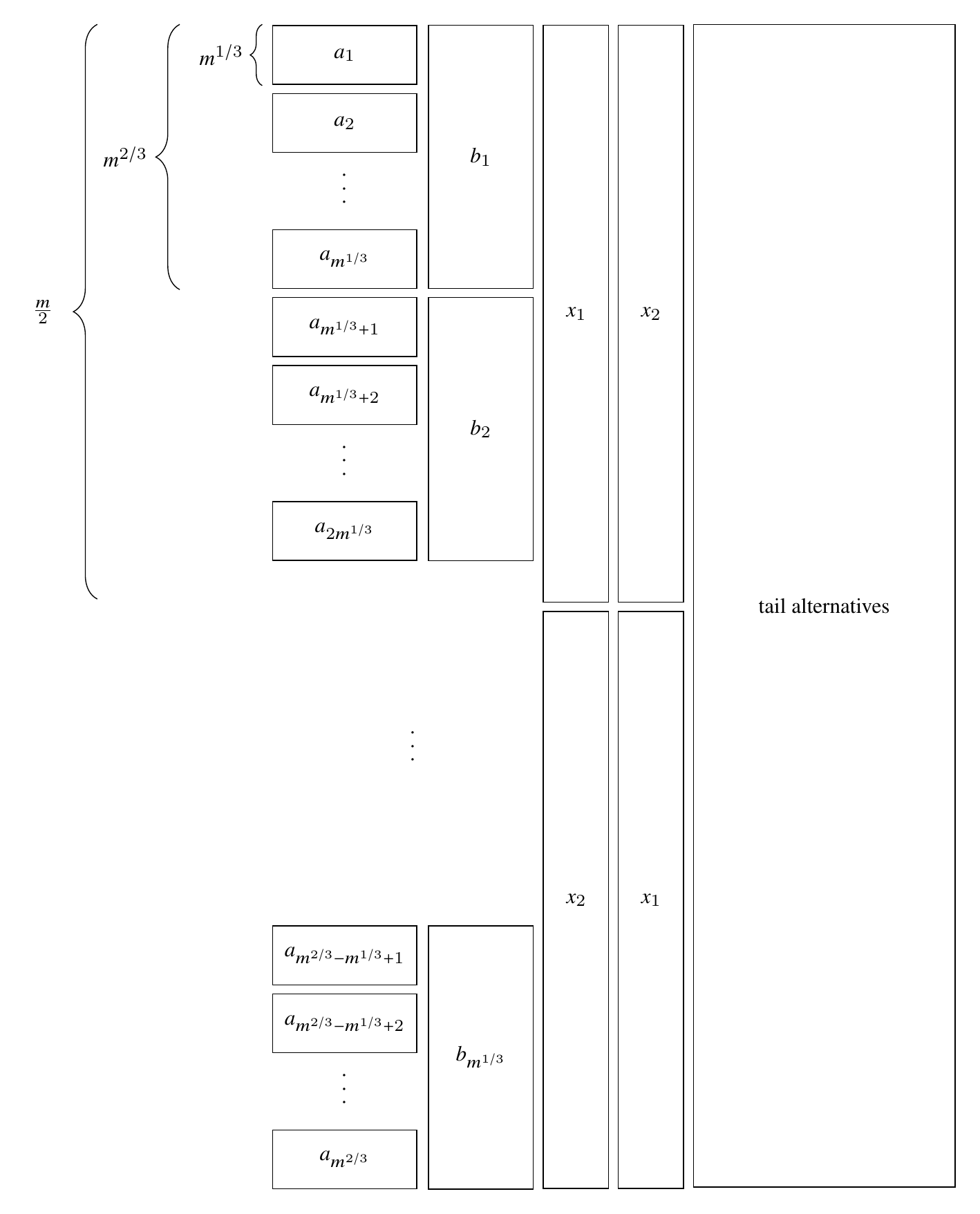}
\caption{An example of the instance used in the proof of \cref{thm:lower-general-unrestricted} for $\lambda=2$; for convenience, we denote here the alternatives of $A_1$ and $A_2$ as $A_1 = \{a_1, ..., a_{m^{2/3}}\}$ and $A_2=\{b_1, ..., b_{m^{1/3}}\}$.}
\label{fig:instance-unrestricted}
\end{figure}

Let $\M$ be an arbitrary mechanism that makes $\lambda$ value queries per agent. Naturally, we assume that $\M$ does not elect any dominated alternative from $Y$. The cardinal information that is revealed due to the queries of $\M$ is as follows:
\begin{itemize}
\item Each of the first $\frac{\lambda}{\lambda+1} \cdot m^{\frac{j}{\lambda+1}}$ queries of $\M$ for any alternative $z \in A_j$, $j \in [\lambda]$ reveals a value of $m^{-\frac{j+1/2}{\lambda+1}}$, while each of the remaining $\frac{1}{\lambda+1} \cdot m^{\frac{j}{\lambda+1}}$ queries of $\M$ for $z$ reveals a value of $m^{-\frac{j}{\lambda+1}}$.

\item Every query of $\M$ for an alternative of $X$ reveals a value of $m^{-1}$.

\item Every query of $\M$ for a tail alternative reveals zero value.  
\end{itemize}
To simplify our notation in the rest of the proof, let $\one_{ij}$ be the indicator variable:
\begin{align*}
\one_{ij} =
\begin{cases}
1, & \text{if $\M$ asks agent $i$ for $a_{ij}$ and has previously asked} \\
   & \text{strictly less than $\frac{\lambda}{\lambda+1} \cdot m^{\frac{1}{\lambda+1}}$ other agents of $T_j(a_{ij})$ for $a_{ij}$} \\
0, & \text{otherwise.}
\end{cases}
\end{align*} 

Now, assume towards a contradiction that $\M$ has distortion $\mathcal{D}(\M) \not\in \Omega\left(\frac{1}{\lambda+1} \cdot m^{\frac{1}{2(\lambda+1)}}\right)$. Using the next two claims, we will show by induction that $\M$ must query a large proportion of the agents at the first $\lambda$ positions, since otherwise the distortion of $\M$ would be $\Omega\left(\frac{1}{\lambda+1} \cdot m^{\frac{1}{2(\lambda+1)}}\right)$.

\begin{claim}\label{claim:general-base}
The mechanism $\M$ must ask at the first position strictly more than $\frac{\lambda}{\lambda+1} \cdot m^{\frac{1}{\lambda+1}}$ of the agents in $T_1(z)$ for every alternative $z \in A_1$.
\end{claim}

\begin{claim}\label{claim:general-induction}
Given that for every alternative $z \in A_j$, $j \in [\lambda-1]$ the mechanism $\M$ asks at the first $j$ positions strictly more than $\left(1-\frac{j}{\lambda+1}\right) \cdot m^{\frac{j}{\lambda+1}}$ of the agents in $T_j(z)$, $\M$ must ask at the first $j+1$ positions strictly more than $\left(1-\frac{j+1}{\lambda+1}\right) \cdot m^{\frac{j+1}{\lambda+1}}$ of the agents in $T_{j+1}(w)$ 
for every alternative $w \in A_{j+1}$.
\end{claim} 

By \cref{claim:general-base} and \cref{claim:general-induction}, $\M$ must ask at the first $\lambda$ positions strictly more than $(1-\frac{\lambda}{\lambda+1}) \cdot m^{\frac{\lambda}{\lambda+1}} = \frac{1}{\lambda+1} \cdot m^{\frac{\lambda}{\lambda+1}} $ of the agents in $T_\lambda(z)$ for every alternative $z \in A_\lambda$. 
Consequently, since $A_\lambda$ consists of exactly $m^{\frac{1}{\lambda+1}}$ alternatives, there are at least $m^{\frac{1}{\lambda+1}} \cdot \frac{1}{\lambda+1} \cdot m^{\frac{\lambda}{\lambda+1}} = \frac{1}{\lambda+1}\cdot m$ agents that are {\em not} queried at positions $(\lambda+1)$ and $(\lambda+2)$ for the alternatives of $X = \{x_1,x_2\}$.
Let $S$ be the set of these $\frac{1}{\lambda+1}\cdot m$ agents; observe that half of them rank $x_1$ ahead of $x_2$ and half of them rank $x_1$ below $x_2$, which follows by the fact that $S$ includes the same number of agents per alternative of $A_\lambda$ and the definition of the ordinal profile. Further, we define two more sets of agents: 
$S_{12} = \{i \in S: i \leq \frac{m}{2}\}$ and $S_{21} = S \setminus S_{12}$. Observe that all agents of $S_{12}$ rank alternative $x_1$ ahead of $x_2$, and all agents of $S_{21}$ rank $x_2$ ahead of $x_1$. 

Now, we define three valuation profiles $\myv_1$, $\myv_2$ and $\myv_3$, which are consistent to the ordinal profile and the cardinal information revealed by the queries of $\M$, but differ on the values that the agents in $S$ have for the alternatives in $X$; in particular, $\myv_2$ and $\myv_3$ are symmetric. 
\begin{itemize}
\item In all three profiles, every agent $i \in [m]$ has value $\one_{ij} \cdot m^{-\frac{j+1/2}{\lambda+1}} + (1-\one_{ij}) \cdot m^{-\frac{j}{\lambda+1}}$ for the alternative $a_{ij} \in A_j$ that she ranks at position $j \in [\lambda]$, and zero value for her tail alternatives;

\item In all three profiles, every agent $i \not\in S$ has value $m^{-1}$ for both $x_1$ and $x_2$;

\item In $\myv_1$, every agent $i \in S$ has value $m^{-\frac{\lambda+1/2}{\lambda+1}}$ for both $x_1$ and $x_2$;

\item In $\myv_2$, every agent $i \in S_{12}$ has value $m^{-1}$ for both $x_1$ and $x_2$, 
while every agent $i \in S_{21}$ has value $m^{-1}$ for $x_1$ and value $m^{-\frac{\lambda+1/2}{\lambda+1}}$ for $x_2$. 

\item In $\myv_3$, every agent $i \in S_{21}$ has value $m^{-1}$ for both $x_1$ and $x_2$, 
while every agent $i \in S_{12}$ has value $m^{-1}$ for $x_2$ and value $m^{-\frac{\lambda+1/2}{\lambda+1}}$ for $x_1$. 
\end{itemize}
Next, we compute the social welfare of each alternative for the different valuation profiles:
\begin{itemize}
\item The social welfare of every alternative $z \in Y$ is
\begin{align*}
\SW(z \,|\, \myv_1) = \SW(z \,|\, \myv_2)  = \SW(z \,|\, \myv_3) = 0.
\end{align*}

\item The social welfare of every alternative $z \in \cup_{j \in [\lambda]} A_j$ is
\begin{align*}
\SW(z \,|\, \myv_1) 
= \SW(z \,|\, \myv_2)  
= \SW(z \,|\, \myv_3) 
&= \frac{\lambda}{\lambda+1} \cdot m^{\frac{j}{\lambda+1}} \cdot m^{-\frac{j+1/2}{\lambda+1}} + \frac{1}{\lambda+1} \cdot m^{\frac{j}{\lambda+1}}  \cdot m^{-\frac{j}{\lambda+1}}  \\
&= \frac{\lambda}{\lambda+1} \cdot m^{-\frac{1}{2(\lambda+1)}} + \frac{1}{\lambda+1}
\leq 1.
\end{align*}

\item The social welfare of $x_1$ is 
\begin{align*}
& \SW(x_1 \,|\, \myv_1) \geq \frac{1}{\lambda+1} \cdot m \cdot m^{-\frac{\lambda+1/2}{\lambda+1}} = \frac{1}{\lambda+1} \cdot m^{\frac{1}{2(\lambda+1)}}, \\
& \SW(x_1 \,|\, \myv_2) = m \cdot m^{-1} =1, \\
& \SW(x_1 \,|\, \myv_3) \geq \frac{1}{2} \cdot \frac{1}{\lambda+1} \cdot m \cdot m^{-\frac{\lambda+1/2}{\lambda+1}} = \frac{1}{2(\lambda+1)} \cdot m^{\frac{1}{2(\lambda+1)}}.
\end{align*}

\item The social welfare of $x_2$ is 
\begin{align*}
& \SW(x_2 \,|\, \myv_1) \geq \frac{1}{\lambda+1} \cdot m \cdot m^{-\frac{\lambda+1/2}{\lambda+1}} = \frac{1}{\lambda+1} \cdot m^{\frac{1}{2(\lambda+1)}}, \\
& \SW(x_2 \,|\, \myv_2) \geq \frac{1}{2}\cdot \frac{1}{\lambda+1} m \cdot m^{-\frac{\lambda+1/2}{\lambda+1}} = \frac{1}{2(\lambda+1)}\cdot m^{\frac{1}{2(\lambda+1)}}, \\
& \SW(x_2 \,|\, \myv_3) = m \cdot m^{-1} =1.
\end{align*}
\end{itemize}

Depending on the choices of the mechanism $\M$, we set the valuation profile to be one of $\myv_1$, $\myv_2$ and $\myv_3$ so that the distortion is as high as possible. In particular, we have:
\begin{itemize}
\item If $\M$ selects any alternative $z \in \cup_{j \in [\lambda]} A_j$, we set the valuation profile to be $\myv_1$. Hence, the social welfare of the winner $z$ is at most $1$, while any alternative of $X$ is optimal with social welfare at least $\frac{1}{\lambda+1} \cdot m^{\frac{1}{2(\lambda+1)}}$, yielding distortion at least $\frac{1}{\lambda+1} \cdot m^{\frac{1}{2(\lambda+1)}}$. 

\item If $\M$ selects alternative $x_1$, we set the valuation profile to be $\myv_2$.
Hence, the social welfare of the winner $x_1$ is exactly $1$, while $x_2$ is the optimal alternative with social welfare at least $\frac{1}{2(\lambda+1)} \cdot m^{\frac{1}{2(\lambda+1)}}$, yielding distortion at least $\frac{1}{2(\lambda+1)} \cdot m^{\frac{1}{2(\lambda+1)}}$.

\item If $\M$ selects alternative $x_2$, we set the valuation profile to be $\myv_3$, which is symmetric to the previous case and again yields distortion at least $\frac{1}{2(\lambda+1)} \cdot m^{\frac{1}{2(\lambda+1)}}$. 
\end{itemize}
Therefore, the distortion of $\M$ is $\Omega\left(\frac{1}{\lambda+1} \cdot m^{\frac{1}{2(\lambda+1)}}\right)$ and the proof of the theorem is now complete; the proofs of \cref{claim:general-base} and \cref{claim:general-induction} can be found in the appendix.
\end{proof}

Using \cref{thm:lower-general-unrestricted}, we can show several lower bounds on the distortion of any mechanism, depending on the number of queries that it makes per agent. In particular, we have the following statement.

\begin{corollary}\label{cor:lower-general-unrestricted}
For unrestricted valuation functions, the distortion of any mechanism $\M$ that uses $\lambda$ queries per agent is
\begin{align*}
\mathcal{D}(\M) = 
\begin{cases}
\Omega\left(m^{\frac{1}{2(\lambda+1)}}\right), & \text{for any constant } \lambda \geq 1 \\
\Omega\left(\log\log{m}\right), & \text{for } \lambda = O\left(\frac{\log{m}}{\log\log{m}}\right).
\end{cases}
\end{align*}
\end{corollary}

\subsection{One-Query Mechanisms with Unit-Sum Valuations}
Next, we turn our attention to unit-sum valuation functions. Coming up with constructions that satisfy the very restricted structure of such valuation functions and at the same time capture {\em all} mechanisms is quite challenging. In the following, we consider mechanisms that are allowed to make only one value query per agent. For this case, we are able to show a weaker lower bound of $\Omega(\sqrt{m}\verythinspace)$, which indicates (but does not prove) some separation between unrestricted and unit-sum valuation functions. 

\begin{theorem}\label{thm:lower-one-unitsum}
For unit-sum valuation functions, the distortion of any mechanism that uses only one value query per agent is $\Omega(\sqrt{m}\verythinspace)$. 
\end{theorem}

\begin{proof}
Consider an instance with $m \geq 4$ alternatives and $n=\sqrt{m}$ agents.
We partition the set $A$ of alternatives into the following four sets:
\begin{itemize}
\item $B = \{b_1, ..., b_{\sqrt{m}}\}$ with $|B| = \sqrt{m}$;
\item $C = \{c_1, c_2\}$ with $|C| = 2$;
\item $D = \{d_1, ..., d_{\sqrt{m}-3}\}$ with $|D| = \sqrt{m}-3$;
\item $E = \{e_1, ..., e_{m-2\sqrt{m}+1}\}$ with $|E| = m - 2\sqrt{m}+1$.
\end{itemize}
Using the notation $[z,w]$ to denote the fact that alternatives $z$ and $w$ are ordered arbitrarily in the ranking of an agent, we define the following ordinal profile:
\begin{itemize}
\item 
The ranking of agent $i \leq \frac{n}{2}$ is $b_i \succ_i c_1 \succ_i c_2 \succ_i [D] \succ_i [E] \succ_i [B \setminus \{b_i\}]$

\item 
The ranking of agent $i > \frac{n}{2}$ is $b_i \succ_i c_2 \succ_i c_1 \succ_i [D] \succ_i [E] \succ_i [B \setminus \{b_i\}]$
\end{itemize} 
Observe that the alternatives of $D \cup E$ are all dominated by alternatives $c_1$ and $c_2$ in the sense that electing $c_1$ or $c_2$ always yields social welfare that is at least as much as the social welfare of any alternative in $D \cup E$. We refer to the alternatives of $E \cup B \setminus \{b_i\}$ as the tail alternatives of agent $i$. 

Let $\M$ be any mechanism that makes one query per agent. 
Naturally, we assume that $\M$ does not elect any dominated alternative from $D \cup E$. 
The possible queries of $\M$ reveal the following cardinal information:
\begin{itemize}
\item A query for the favorite alternative of an agent (at the first position) reveals a value of $\frac{1}{\sqrt{m}}$;
\item A query for an alternative in $C \cup D$ reveals a value of $\frac{1}{m}$;
\item Any other query (for a tail alternative) reveals zero value.
\end{itemize}
We define the following sets of agents, depending on the function of $\M$:
\begin{itemize}
\item $S_1$ is the set of agents queried at the first position (for their favorite alternative);
\item $S_C$ is the set of agents queried for some alternative in $C$; 
\item $S_D$ is the set of agents queried for some alternative in $D$;
\item $S_>$ is the set of agents queried for some tail alternative.
\end{itemize}
Next, we distinguish between three cases, depending on the alternative that $\M$ elects.

\paragraph{Case I: $\M$ selects alternative $c_1$ (the case of $c_2$ is symmetric)} \ \\
\noindent 
If $|S_1| < n$, we define the following valuation profile $\myv$:
\begin{itemize}
\item For every agent $i \in S_C$, we set $v_{i,b_i} = 1- \frac{2}{m}$ and $v_{i,c_1}=v_{i,c_2}= \frac{1}{m}$; 
the value for all other alternatives is zero.

\item For every agent $i \in S_D$, we set $v_{i,b_i} = 1 -\frac{\sqrt{m}-1}{m}$ and
$v_{i,c_1}= v_{i,c_2}= v_{i,d_j} = \frac{1}{m}$ for $j \in [\sqrt{m}-3]$;
the value for all other alternatives is zero.  

\item For every agent $i \in S_>$, we set $v_{i,b_i} = 1$;
the value for all other alternatives is zero. 

\item For every agent $i \in S_1$, we set $v_{i,b_i} = \frac{1}{\sqrt{m}}$, and split the remaining value of $1-\frac{1}{\sqrt{m}}$ equally among all other $m-1$ alternatives so that for each of them the value of agent $i$ is $\frac{\sqrt{m}-1}{\sqrt{m}(m-1)}$.  
\end{itemize}
Hence, alternative $c_1$ has social welfare
\begin{align*}
\SW(c_1 \,|\, \myv) &= \bigg( |S_{C}| + |S_{D}| \bigg) \cdot \frac{1}{m} + |S_>| \cdot 0 + |S_1| \cdot \frac{\sqrt{m}-1}{\sqrt{m}(m-1)} \\
&\leq \bigg(|S_1| + |S_{C}| + |S_{D}| \bigg) \cdot  \frac{1}{m} \\
&\leq \frac{1}{\sqrt{m}} \,,
\end{align*}
where the first inequality follows since $\frac{\sqrt{m}-1}{\sqrt{m}(m-1)} \leq \frac{1}{m} \Leftrightarrow \sqrt{m} \leq m$, and the second follows by the fact that $|S_1| + |S_{C}| + |S_{D}| \leq n = \sqrt{m}$.  
Since $|S_1| < n$, there exists an agent $i^* \in S_{C} \cup S_{D} \cup S_>$ such that her favorite alternative $b_{i^*}$ has social welfare 
\begin{align*}
\SW(b_{i^*} \,|\, \myv) \geq 1 -\frac{\sqrt{m}-1}{m} \geq 1 - \frac{1}{\sqrt{m}} \,.
\end{align*}
As a result, the distortion is at least $\sqrt{m} - 1$.

If $|S_1| = n$, we define the following valuation profile $\myv$:
\begin{itemize}
\item For every agent $i \leq \frac{n}{2}$, we set $v_{i,b_i}=\frac{1}{\sqrt{m}}$, and split the remaining value of $1-\frac{1}{\sqrt{m}}$ equally among all other $m-1$ alternatives so that for each of them the value of agent $i$ is $\frac{\sqrt{m}-1}{\sqrt{m}(m-1)}$. 

\item For every agent $i > \frac{n}{2}$, we set $v_{i,b_i} = v_{i,c_2} = \frac{1}{\sqrt{m}}$, and split the remaining value of $1-\frac{2}{\sqrt{m}}$ equally among all other $m-2$ alternatives so that for each of them the value of agent $i$ is $\frac{\sqrt{m}-2}{\sqrt{m}(m-2)}$.
\end{itemize}
Hence, alternative $c_1$ has social welfare
\begin{align*}
\SW(c_1 \,|\, \myv) &= \frac{n}{2} \cdot \left( \frac{\sqrt{m}-1}{\sqrt{m}(m-1)} + \frac{\sqrt{m}-2}{\sqrt{m}(m-2)}  \right) \\
&\leq \frac{n}{2} \cdot \frac{2\sqrt{m}-3}{\sqrt{m}(m-2)}.
\end{align*}
On the other hand, alternative $c_2$ has social welfare
\begin{align*}
\SW(c_2 \,|\, \myv) &= \frac{n}{2} \cdot \left( \frac{\sqrt{m}-1}{\sqrt{m}(m-1)} + \frac{1}{\sqrt{m}}  \right) \\
&= \frac{n}{2} \cdot \frac{m + \sqrt{m}-2}{\sqrt{m}(m-1)}.
\end{align*}
Consequently, the distortion is at least
\begin{align*}
\frac{\SW(c_2 \,|\, \myv)}{\SW(c_1 \,|\, \myv)} \geq \frac{m-2}{m-1} \cdot \frac{m+\sqrt{m}-2}{2\sqrt{m}-3} \geq \frac{1}{2} \sqrt{m},
\end{align*}
where the last inequality holds for any $m \geq 3$.

\paragraph{Case II: $|S_1| \geq 1$ and $\M$ selects some alternative $b_{i^*}$ for $i^* \in S_1$.} \ \\
\noindent
If $|S_1| < n$, we define the following valuation profile $\myv$:
\begin{itemize}
\item For every agent $i \in S_{C}$, we set $v_{i,b_i} = 1- \frac{2}{m}$ and $v_{i,c_1}=v_{i,c_2}= \frac{1}{m}$;
the value for all other alternatives is zero.

\item For every agent $i \in S_{D}$, we set $v_{i,b_i} = 1 - \frac{1}{\sqrt{m}} + \frac{1}{m}$ and
$v_{i,c_1}= v_{i,c_2}= v_{i,d_j} = \frac{1}{m}$ for $j \in [\sqrt{m}-3]$;
the value for all other alternatives is zero.

\item For every $i \in S_>$, we set $v_{i,b_i} = 1$;
the value for all other alternatives is zero.

\item For every $i \in S_1$, we set $v_{i,b_i} = \frac{1}{\sqrt{m}}$, and split the remaining value of $1-\frac{1}{\sqrt{m}}$ equally among the $m - \sqrt{m}$ alternatives of $C \cup D \cup E$, while the value for the alternatives of $B \setminus \{b_i\}$ is zero. 
\end{itemize}
Hence, the social welfare of alternative $b_{i^*}$ is 
$$\SW(b_{i^*} \,|\, \myv) = \frac{1}{\sqrt{m}}.$$ 
Since $|S_1| < n$, there exists an agent $i \in S_{C} \cup S_{D} \cup S_>$ such that alternative $b_i$ has social welfare 
$$ \SW(b_i \,|\, \myv) \geq 1 - \frac{1}{\sqrt{m}} + \frac{1}{m} \geq 1 - \frac{1}{\sqrt{m}},$$ 
and therefore the distortion is at least $\sqrt{m} - 1$.

If $|S_1| = n$, we define the following valuation profile $\myv$:
\begin{itemize}
\item For every agent $i \in [n]$, we set $v_{i,b_i} = v_{i,c_1}=v_{i,c_2}= \frac{1}{\sqrt{m}}$, and split the remaining value of $1-\frac{3}{\sqrt{m}}$ equally among the $m - \sqrt{m}-2$ alternatives of $D \cup E$; the value for the alternatives of $B \setminus \{b_i\}$ is zero. 
This is a valid valuation definition since the value for each alternative in $D \cup E$ is $\frac{\sqrt{m}-3}{\sqrt{m}(m-\sqrt{m}-2)} \leq \frac{1}{\sqrt{m}} \Leftrightarrow (\sqrt{m}-1)^2\geq 0$.
\end{itemize}
Hence, alternative $b_{i^*}$ has social welfare
$$\SW(b_{i^*} \,|\, \myv) = \frac{1}{\sqrt{m}}.$$ 
But now, the social welfare of $c_1$ and $c_2$ is equal to 
$$\SW(c_1 \,|\, \myv) = \SW(c_2 \,|\, \myv) = n \frac{1}{\sqrt{m}} = 1,$$ 
yielding distortion that is at least $\sqrt{m}$.

\paragraph{Case III: $|S_C| + |S_D| + |S_>|  \geq 1$ and $\M$ selects some alternative $b_{i^*}$ for $i^* \in S_C \cup S_D \cup S_>$} \ \\
\noindent
We define the following valuation profile $\myv$:
\begin{itemize}
\item If $i^* \in S_C \cup S_D$, we set $v_{i^*j} = \frac{1}{m}$ for every alternative $j \in A$. 
If $i^* \in S_>$, we split the total value of $1$ equally among the $m - \sqrt{m} + 1$ alternatives in $\{b_{i^*}\} \cup C \cup D \cup E$ so that the value of agent $i$ for each such alternative is $\frac{1}{m-\sqrt{m}+1} \leq \frac{1}{\sqrt{m}}$. 

\item For every agent $i \in S_C \setminus \{i^*\}$, we set $v_{i,b_i} = 1- \frac{2}{m}$ and $v_{i,c_1}=v_{i,c_2}= \frac{1}{m}$;
the value for all other alternatives is zero.

\item For every agent $i \in S_D \setminus \{i^*\}$, we set $v_{i,b_i} = 1 -\frac{1}{\sqrt{m}} + \frac{1}{m}$ and
$v_{i,c_1}= v_{i,c_2}= v_{i,d_j} = \frac{1}{m}$ for $j \in [\sqrt{m}-3]$;
the value for all other alternatives is zero.

\item For every agent $i \in S_> \setminus \{i^*\}$, we set $v_{i,b_i} = 1$;
the value for all other alternatives is zero.

\item For every agent $i \in S_1$, we set $v_{i,b_i} = v_{i,c_1}=v_{i,c_2}= \frac{1}{\sqrt{m}}$, and split the remaining value of $1-\frac{3}{\sqrt{m}}$ equally among the $m-\sqrt{m}-2$ alternatives in $D \cup E$; 
the value of agent $i$ for the alternatives of $B \setminus \{b_i\}$ is zero.
This is a valid valuation definition since the value for each alternative in $D \cup E$ is $\frac{1 - \frac{3}{\sqrt{m}}}{m-2\sqrt{m}+1} \leq \frac{1}{\sqrt{m}} \Leftrightarrow (\sqrt{m}-1)^2 \geq 0$.
\end{itemize}
Hence, in any case, the social welfare of alternative $b_{i^*}$ is
\begin{align*}
\SW(b_{i^*} \,|\, \myv) \leq \frac{1}{\sqrt{m}}.
\end{align*}
We now distinguish between a couple more cases:
\begin{itemize}
\item If $|S_C| + |S_D| + |S_>| \geq 2$, then there exists an agent $i \in S_C \cup S_D \cup S_> \setminus \{i^*\}$ such that the social welfare of alternative $b_i$ is 
\begin{align*}
\SW(b_i \,|\, \myv) \geq 1 -\frac{1}{\sqrt{m}} + \frac{1}{m} \geq 1 -\frac{1}{\sqrt{m}},
\end{align*}
yielding distortion at least $\sqrt{m}-1$.
 
\item If $|S_C| + |S_D| + |S_>| = 1$, then since $|S_1|= n-1 = \sqrt{m}-1$, alternatives $c_1$ and $c_2$ both have social welfare 
\begin{align*}
\SW(c_1 \,|\, \myv) = \SW(c_2 \,|\, \myv) \geq (\sqrt{m}-1)\frac{1}{\sqrt{m}} = 1 - \frac{1}{\sqrt{m}}
\end{align*}
and the distortion is at least $\sqrt{m} - 1$.
\end{itemize}
The proof is now complete.
\end{proof}

\section{Conclusions and Directions for Future Research}\label{sec:conclusions}
We studied mechanisms for general single winner elections. In particular, we explored the potential of improving the distortion of deterministic ordinal mechanisms by making a \emph{limited} number of cardinal queries per agent. On this front, we obtained a definitive positive answer. As highlights of our positive results, we showed that it is possible to achieve \emph{constant} distortion by making $O(\log^2{m})$ value or comparison queries per agent, while only $O(\log m)$ value queries are enough to guarantee distortion $O(\sqrt{m})$, thus outperforming the best known randomized ordinal mechanism of~\cite{boutilier2015optimal}. Quite interestingly, our positive results for value queries hold without any normalization assumptions, which makes them even stronger.

We complemented these results by showing (nearly) tight lower bounds for many interesting cases. For one-query mechanisms we showed a linear lower bound for unrestricted valuation functions and a lower bound of $\Omega(\sqrt{m}\verythinspace)$ for unit-sum valuations. Further, for mechanisms that make $O\left(\frac{\log{m}}{\log\log{m}}\right)$ queries we showed a superconstant lower bound for unrestricted valuation functions.  

Possibly the most obvious open problem is to fill in the gaps between our upper and lower bounds. To this end, we make the following two conjectures.

\medskip

\noindent
{\bf $O(1)$\textbf{-Query Conjecture}.}
{\em There exists a mechanism that achieves a distortion of $O(\sqrt{m}\verythinspace)$ using a constant number of value queries per agent, for unit-sum or unrestricted valuation functions.}

\medskip

\noindent
{\bf $(\log{m})$\textbf{-Queries Conjecture}.}
{\em There exists a mechanism that achieves a constant distortion, using $O(\log{m})$ value queries per agent, for unit-sum or unrestricted valuation functions.}

\medskip

\noindent
We consider settling these two conjectures the most interesting problems left open in our work. Since our upper bounds for value queries do not make use of the unit-sum normalization, it is conceivable that some clever use of that extra information could possibly lead to better trade-offs.

A natural direction for future work is to consider randomization. Intriguingly, one could consider two different levels of randomization. The first level consists of mechanisms that decide randomly what queries to make to the agents, yet the winning alternative is chosen deterministically. The second level consists of mechanisms that use randomization for both querying and making the final decision. Both of these two classes of randomized mechanisms are very natural and may lead to similar distortion bounds but potentially using fewer queries. 

Our work takes a first step towards exploring how powerful ordinal mechanisms with limited access to cardinal information can actually be. Of course, the same idea can be applied to many different contexts, such as participatory budgeting, multi-winner elections, or the metric distortion setting, which has been extensively studied over the past years. As we mentioned in the introduction, \citet{abramowitz2019awareness} already take a step in this direction in the metric setting.

\appendix

\section{Missing Proofs from Subsection \ref{sec:lower-bounds_lambda}}

\subsection*{Proof of \cref{claim:general-base}}
Assume towards a contradiction that there exists an alternative $z^* \in A_1$ such that the mechanism $\M$ asks at most $\frac{\lambda}{\lambda+1} \cdot m^{\frac{1}{\lambda+1}}$ agents of $T_1(z^*)$ at the first position, and $\M$ has distortion $\mathcal{D}(\M) \not\in \Omega\left(\frac{1}{\lambda+1} \cdot m^{\frac{1}{2(\lambda+1)}}\right)$. Let $S$ be the set of the at least $\frac{1}{\lambda+1} \cdot m^{\frac{1}{\lambda+1}}$ agents of $T_1(z^*)$ that are {\em not} queried by $\M$ at the first position. Hence, we have that $\one_{i1}=1$ for every agent $i \not\in S$.

We now define two valuation profiles $\myv_1$ and $\myv_2$, which are consistent to the ordinal profile and any information revealed by the queries of $\M$, but differ on the value that the agents of $S$ have for alternative $z^*$:
\begin{itemize}
\item In both $\myv_1$ and $\myv_2$, every agent $i \not\in S$ has value $\one_{ij} \cdot m^{-\frac{j+1/2}{\lambda+1}} + (1-\one_{ij}) \cdot m^{-\frac{j}{\lambda+1}}$ for the alternative $a_{ij} \in A_j$ that she ranks at position $j \in [\lambda]$, value $m^{-1}$ for alternatives $x_1$ and $x_2$, and zero value for her tail alternatives.

\item In both $\myv_1$ and $\myv_2$, every agent $i \in S$ has value $\one_{ij} \cdot m^{-\frac{j+1/2}{\lambda+1}} + (1-\one_{ij}) \cdot m^{-\frac{j}{\lambda+1}}$ for the alternative $a_{ij} \in A_j$ that she ranks at position $j \in [\lambda] \setminus \{1\}$, value $m^{-1}$ for alternatives $x_1$ and $x_2$, and zero value for her tail alternatives. 

\item In $\myv_1$, every agent $i \in S$ has value $m^{-\frac{3/2}{\lambda+1}}$ for $z^*$.
\item In $\myv_2$, every agent $i \in S$ has value $1$ for $z^*$.
\end{itemize}
Given the definition of these two valuation profiles, it is easy to compute the social welfare of the alternatives:
\begin{itemize}
\item The social welfare of every alternative $z \in Y$ is
\begin{align*}
\SW(z \,|\, \myv_1) = \SW(z \,|\, \myv_2) = 0.
\end{align*}

\item The social welfare of alternatives $x_1$ and $x_2$ is
\begin{align*}
\SW(x_1 \,|\, \myv_1) = \SW(x_1 \,|\, \myv_2) = \SW(x_2 \,|\, \myv_1) = \SW(x_2 \,|\, \myv_2) = m \cdot m^{-1}=1.
\end{align*}

\item The social welfare of any alternative $z \in \cup_{j \in [\lambda]} A_j \setminus \{z^*\}$ is 
\begin{align*}
\SW(z \,|\, \myv_1) = \SW(z \,|\, \myv_2) 
&= \frac{\lambda}{\lambda+1} \cdot m^{\frac{j}{\lambda+1}} \cdot m^{-\frac{j+1/2}{\lambda+1}} + \frac{1}{\lambda+1} \cdot m^{\frac{j}{\lambda+1}}  \cdot m^{-\frac{j}{\lambda+1}} \\
&= \frac{\lambda}{\lambda+1} \cdot m^{-\frac{1}{2(\lambda+1)}} + \frac{1}{\lambda+1}
\leq 1.
\end{align*}

\item The social welfare of alternative $z^*$ is
\begin{align*}
&\SW(z^* \,|\, \myv_1) = \frac{\lambda}{\lambda+1} \cdot m^{\frac{1}{\lambda+1}} \cdot m^{-\frac{1 + 1/2}{\lambda+1}} + \frac{1}{\lambda+1} \cdot m^{\frac{1}{\lambda+1}} \cdot m^{-\frac{3/2}{\lambda+1}} = m^{-\frac{1}{2(\lambda+1)}}, \\
&\SW(z^* \,|\, \myv_2) = \frac{\lambda}{\lambda+1} \cdot m^{\frac{1}{\lambda+1}} \cdot m^{-\frac{1+1/2}{\lambda+1}} + \frac{1}{\lambda+1} \cdot m^{\frac{1}{\lambda+1}} \cdot 1
= \frac{\lambda}{\lambda+1} \cdot m^{-\frac{1}{2(\lambda+1)}} + \frac{1}{\lambda+1} \cdot m^{\frac{1}{\lambda+1}}.
\end{align*}
\end{itemize}

Depending on the choices of the mechanism $\M$, we set the valuation profile to be either $\myv_1$ or $\myv_2$ so that the distortion is as high as possible. In particular, we have:
\begin{itemize}
\item If $\M$ selects alternative $z^*$, we set the valuation profile to be $\myv_1$. Hence, the social welfare of the winner $z^*$ is $m^{-\frac{1}{2} \cdot \frac{1}{\lambda+1}}$, while any alternative of $X$ is optimal with social welfare $1$, yielding distortion $m^{\frac{1}{2(\lambda+1)}}$.

\item If $\M$ selects some alternative $z \in X \cup_{j \in [\lambda]} A_j \setminus \{z^*\}$, we set the valuation profile to be $\myv_2$. Hence, the social welfare of the winner $z$ is at most $1$, while alternative $z^*$ is optimal with social welfare at least $\frac{1}{\lambda+1} \cdot m^{\frac{1}{\lambda+1}}$, yielding distortion at least $\frac{1}{\lambda+1} \cdot m^{\frac{1}{\lambda+1}}$.
\end{itemize}
In any case, the distortion is $\Omega\left(\frac{1}{\lambda+1} \cdot m^{\frac{1}{2(\lambda+1)}}\right)$ and the proof of the claim follows.
\hfill 
$\qed$
\medskip

\subsection*{Proof of \cref{claim:general-induction}}
For every alternative $w \in A_{j+1}$, let $S_w \subseteq T_{j+1}(w)$ be the set of agents that rank $w$ at position $(j+1)$ and are queried by $\M$ at the first $j$ positions. 
By the definition of the ordinal profile, the set $T_{j+1}(w)$ consists of $m^{\frac{j+1}{\lambda+1}}$ agents. These agents are partitioned into $m^{\frac{1}{\lambda+1}}$ sets of size $m^{\frac{j}{\lambda+1}}$ so that the agents of each such set all rank the same alternative of $A_j$ at position $j$.
Therefore, by the assumption of the claim, we have that $|S_w| > m^{\frac{1}{\lambda+1}} \cdot \left(1-\frac{j}{\lambda+1}\right)\cdot m^{\frac{j}{\lambda+1}} = \left(1-\frac{j}{\lambda+1}\right) \cdot m^{\frac{j+1}{\lambda+1}}$.

Now, assume towards a contradiction that there exists an alternative $w^* \in A_{j+1}$ such that $\M$ queries at most $(1-\frac{j+1}{\lambda+1})\cdot m^{\frac{j+1}{\lambda+1}}$ of the agents in $S_{w^*}$ at the first $(j+1)$ positions, and $\mathcal{D}(\M) \not\in \Omega\left(\frac{1}{\lambda+1} \cdot m^{\frac{1}{2(\lambda+1)}}\right)$. 
Let $S$ be the set of the agents in $S_{w^*}$ that are not queried by $\M$ at position $(j+1)$.  
By our discussion so far, we have that $|S| \geq |S_{w^*}| - \left(1-\frac{j+1}{\lambda+1} \right)\cdot m^{\frac{j+1}{\lambda+1}} > \frac{1}{\lambda+1} \cdot m^{\frac{j+1}{\lambda+1}}$, and therefore $\one_{i, j+1}=1$ for every agent $i \not\in S$.

We now define two valuation profiles $\myv_1$ and $\myv_2$, which are consistent to the ordinal profile and any information revealed by the queries of $\M$, but differ on the value that the agents of $S$ have for alternative $w^*$:
\begin{itemize}
\item In both $\myv_1$ and $\myv_2$, every agent $i \not\in S$ has value $\one_{i\ell} \cdot m^{-\frac{\ell+1/2}{\lambda+1}} + (1-\one_{i\ell}) \cdot m^{-\frac{\ell}{\lambda+1}}$ for the alternative $a_{i\ell} \in A_\ell$ that she ranks at position $\ell \in [\lambda]$, value $m^{-1}$ for alternatives $x_1$ and $x_2$, and zero value for her tail alternatives.

\item In both $\myv_1$ and $\myv_2$, every agent $i \in S$ has value $\one_{i\ell} \cdot m^{-\frac{\ell+1/2}{\lambda+1}} + (1-\one_{i\ell}) \cdot m^{-\frac{\ell}{\lambda+1}}$ for the alternative $a_{i\ell} \in A_\ell$ that she ranks at position $\ell \in [\lambda] \setminus \{j+1\}$, value $m^{-1}$ for alternatives $x_1$ and $x_2$, and zero value for her tail alternatives. 

\item In $\myv_1$, every agent $i \in S$ has value $m^{-\frac{j+3/2}{\lambda+1}}$ for $w^*$.
\item In $\myv_2$, every agent $i \in S$ has value $m^{-\frac{j+1/2}{\lambda+1}}$ for $w^*$.
\end{itemize}
Given the definition of these two valuation profiles, it is easy to compute the social welfare of the alternatives:
\begin{itemize}
\item The social welfare of every alternative $z \in Y$ is
\begin{align*}
\SW(z \,|\, \myv_1) = \SW(z \,|\, \myv_2) = 0.
\end{align*}

\item The social welfare of alternatives $x_1$ and $x_2$ is
\begin{align*}
\SW(x_1 \,|\, \myv_1) = \SW(x_1 \,|\, \myv_2) = \SW(x_2 \,|\, \myv_1) = \SW(x_2 \,|\, \myv_2) = m \cdot m^{-1}=1.
\end{align*}

\item The social welfare of any alternative $z \in \cup_{\ell \in [\lambda]} A_\ell \setminus \{w^*\}$ is 
\begin{align*}
\SW(z \,|\, \myv_1) = \SW(z \,|\, \myv_2) 
&= \frac{\lambda}{\lambda+1} \cdot m^{\frac{\ell}{\lambda+1}} \cdot m^{-\frac{\ell+1/2}{\lambda+1}} + \frac{1}{\lambda+1} \cdot m^{\frac{\ell}{\lambda+1}}  \cdot m^{-\frac{\ell}{\lambda+1}} \\
&= \frac{\lambda}{\lambda+1} \cdot m^{-\frac{1}{2(\lambda+1)}} + \frac{1}{\lambda+1}
\leq 1.
\end{align*}

\item The social welfare of alternative $w^*$ is
\begin{align*}
\SW(w^* \,|\, \myv_1) & = \frac{\lambda}{\lambda+1} \cdot m^{\frac{j+1}{\lambda+1}} \cdot m^{-\frac{j+1+1/2}{\lambda+1}} + \frac{1}{\lambda+1} \cdot m^{\frac{j+1}{\lambda+1}} \cdot m^{-\frac{j+3/2}{\lambda+1}} = m^{-\frac{1}{2(\lambda+1)}}, \\
\SW(w^* \,|\, \myv_2) &= \frac{\lambda}{\lambda+1} \cdot m^{\frac{j+1}{\lambda+1}} \cdot m^{-\frac{j+1 + 1/2}{\lambda+1}} + \frac{1}{\lambda+1} \cdot m^{\frac{j+1}{\lambda+1}} \cdot m^{-\frac{j+1/2}{\lambda+1}} \\
&= \frac{\lambda}{\lambda+1} \cdot m^{-\frac{1}{2(\lambda+1)}} + \frac{1}{\lambda+1} \cdot m^{\frac{1}{2(\lambda+1)}}.
\end{align*}
\end{itemize}

Depending on the choices of the mechanism $\M$, we set the valuation profile to be either $\myv_1$ or $\myv_2$ so that the distortion is as high as possible. In particular, we have:
\begin{itemize}
\item If $\M$ selects alternative $w^*$, we set the valuation profile to be $\myv_1$. Hence, the social welfare of the winner $w^*$ is $m^{-\frac{1}{2(\lambda+1)}}$, while any alternative of $X$ is optimal with social welfare $1$, yielding distortion $m^{\frac{1}{2(\lambda+1)}}$.

\item If $\M$ selects some alternative $z \in X \cup_{\ell \in [\lambda]} A_\ell \setminus \{w^*\}$, we set the valuation profile to be $\myv_2$. Hence, the social welfare of the winner $z$ is at most $1$, while alternative $w^*$ is optimal with social welfare at least $\frac{1}{\lambda+1} \cdot m^{\frac{1}{2(\lambda+1)}}$, yielding distortion at least $\frac{1}{\lambda+1} \cdot m^{\frac{1}{2(\lambda+1)}}$.
\end{itemize}
In any case, the distortion is $\Omega\left(\frac{1}{\lambda+1} \cdot m^{\frac{1}{2(\lambda+1)}}\right)$ and the proof of the claim follows.
\hfill 
$\qed$

\bibliographystyle{plainnat}
\bibliography{references}

\begin{thebibliography}{44}
\providecommand{\natexlab}[1]{#1}
\providecommand{\url}[1]{\texttt{#1}}
\expandafter\ifx\csname urlstyle\endcsname\relax
  \providecommand{\doi}[1]{doi: #1}\else
  \providecommand{\doi}{doi: \begingroup \urlstyle{rm}\Url}\fi

\bibitem[Abramowitz and Anshelevich(2018)]{abramowitz2017utilitarians}
Ben Abramowitz and Elliot Anshelevich.
\newblock Utilitarians without utilities: Maximizing social welfare for graph
  problems using only ordinal preferences.
\newblock In \emph{Proceedings of the 32nd {AAAI} Conference on Artificial
  Intelligence ({AAAI})}, pages 894--901, 2018.

\bibitem[Abramowitz et~al.(2019)Abramowitz, Anshelevich, and
  Zhu]{abramowitz2019awareness}
Ben Abramowitz, Elliot Anshelevich, and Wennan Zhu.
\newblock Awareness of voter passion greatly improves the distortion of metric
  social choice.
\newblock In \emph{Proceedings of the The 15th Conference on Web and Internet
  Economics (WINE)}, pages 3--16, 2019.

\bibitem[Amanatidis et~al.(2020)Amanatidis, Birmpas, Filos{-}Ratsikas, and
  Voudouris]{conference}
Georgios Amanatidis, Georgios Birmpas, Aris Filos{-}Ratsikas, and Alexandros~A.
  Voudouris.
\newblock Peeking behind the ordinal curtain: Improving distortion via cardinal
  queries.
\newblock In \emph{Proceedings of the 34th {AAAI} Conference on Artificial
  Intelligence ({AAAI})}, pages 1782--1789, 2020.

\bibitem[Amanatidis et~al.(2021)Amanatidis, Birmpas, Filos-Ratsikas, and
  Voudouris]{amanatidis2020few}
Georgios Amanatidis, Georgios Birmpas, Aris Filos-Ratsikas, and Alexandros~A
  Voudouris.
\newblock A few queries go a long way: Information-distortion tradeoffs in
  matching.
\newblock In \emph{Proceedings of the 35th AAAI Conference on Artificial
  Intelligence ({AAAI})}, 2021.

\bibitem[Anshelevich and Postl(2017)]{anshelevich2017randomized}
Elliot Anshelevich and John Postl.
\newblock Randomized social choice functions under metric preferences.
\newblock \emph{Journal of Artificial Intelligence Research}, 58:\penalty0
  797--827, 2017.

\bibitem[Anshelevich and Sekar(2016)]{anshelevich2016blind}
Elliot Anshelevich and Shreyas Sekar.
\newblock Blind, greedy, and random: Algorithms for matching and clustering
  using only ordinal information.
\newblock In \emph{Proceedings of the 30th {AAAI} Conference on Artificial
  Intelligence ({AAAI})}, pages 390--396, 2016.

\bibitem[Anshelevich and Zhu(2017)]{anshelevich2017tradeoffs}
Elliot Anshelevich and Wennan Zhu.
\newblock Tradeoffs between information and ordinal approximation for bipartite
  matching.
\newblock In \emph{Proceedings of the 10th International Symposium on
  Algorithmic Game Theory ({SAGT})}, pages 267--279, 2017.

\bibitem[Anshelevich and Zhu(2018)]{anshelevich2018ordinal}
Elliot Anshelevich and Wennan Zhu.
\newblock Ordinal approximation for social choice, matching, and facility
  location problems given candidate positions.
\newblock In \emph{Proceedings of the 14th International Conference on Web and
  Internet Economics ({WINE})}, pages 3--20, 2018.

\bibitem[Anshelevich et~al.(2018)Anshelevich, Bhardwaj, Elkind, Postl, and
  Skowron]{anshelevich2018approximating}
Elliot Anshelevich, Onkar Bhardwaj, Edith Elkind, John Postl, and Piotr
  Skowron.
\newblock Approximating optimal social choice under metric preferences.
\newblock \emph{Artificial Intelligence}, 264:\penalty0 27--51, 2018.

\bibitem[Barbera et~al.(1998)Barbera, Bogomolnaia, and van~der Stel]{Barbera98}
Salvador Barbera, Anna Bogomolnaia, and Hans van~der Stel.
\newblock Strategy-proof probabilistic rules for expected utility maximizers.
\newblock \emph{Mathematical Social Sciences}, 35\penalty0 (2):\penalty0
  89--103, 1998.

\bibitem[Benade et~al.(2017)Benade, Nath, Procaccia, and
  Shah]{benade2017preference}
Gerdus Benade, Swaprava Nath, Ariel~D. Procaccia, and Nisarg Shah.
\newblock Preference elicitation for participatory budgeting.
\newblock In \emph{Proceedings of the 31st {AAAI} Conference on Artificial
  Intelligence ({AAAI})}, pages 376--382, 2017.

\bibitem[Benade et~al.(2019)Benade, Procaccia, and Qiao]{benade2019low}
Gerdus Benade, Ariel~D. Procaccia, and Mingda Qiao.
\newblock Low-distortion social welfare functions.
\newblock In \emph{Proceedings of the 33rd AAAI Conference on Artificial
  Intelligence ({AAAI})}, pages 1788--1795, 2019.

\bibitem[Bhaskar et~al.(2018)Bhaskar, Dani, and Ghosh]{bhaskar2018truthful}
Umang Bhaskar, Varsha Dani, and Abheek Ghosh.
\newblock Truthful and near-optimal mechanisms for welfare maximization in
  multi-winner elections.
\newblock In \emph{Proceedings of the 32nd {AAAI} Conference on Artificial
  Intelligence ({AAAI})}, pages 925--932, 2018.

\bibitem[Bogomolnaia and Moulin(2001)]{BM:01}
Anna Bogomolnaia and Herv\'{e} Moulin.
\newblock {A new solution to the random assignment problem}.
\newblock \emph{Journal of Economic Theory}, 100:\penalty0 295--328, 2001.

\bibitem[Borodin et~al.(2019)Borodin, Lev, Shah, and
  Strangway]{borodin2019primarily}
Allan Borodin, Omer Lev, Nisarg Shah, and Tyrone Strangway.
\newblock Primarily about primaries.
\newblock In \emph{Proceedings of the 33rd AAAI Conference on Artificial
  Intelligence ({AAAI})}, 2019.

\bibitem[Boutilier et~al.(2015)Boutilier, Caragiannis, Haber, Lu, Procaccia,
  and Sheffet]{boutilier2015optimal}
Craig Boutilier, Ioannis Caragiannis, Simi Haber, Tyler Lu, Ariel~D. Procaccia,
  and Or~Sheffet.
\newblock Optimal social choice functions: A utilitarian view.
\newblock \emph{Artificial Intelligence}, 227:\penalty0 190--213, 2015.

\bibitem[Brandt et~al.(2016)Brandt, Conitzer, Endriss, Lang, and
  Procaccia]{comsocbook2016}
Felix Brandt, Vincent Conitzer, Ulle Endriss, J{\'{e}}r{\^{o}}me Lang, and
  Ariel~D. Procaccia, editors.
\newblock \emph{Handbook of Computational Social Choice}.
\newblock Cambridge University Press, 2016.

\bibitem[Caragiannis and Procaccia(2011)]{caragiannis2011embedding}
Ioannis Caragiannis and Ariel~D. Procaccia.
\newblock Voting almost maximizes social welfare despite limited communication.
\newblock \emph{Artificial Intelligence}, 175\penalty0 (9-10):\penalty0
  1655--1671, 2011.

\bibitem[Caragiannis et~al.(2017)Caragiannis, Nath, Procaccia, and
  Shah]{caragiannis2017subset}
Ioannis Caragiannis, Swaprava Nath, Ariel~D. Procaccia, and Nisarg Shah.
\newblock Subset selection via implicit utilitarian voting.
\newblock \emph{Journal of Artificial Intelligence Research}, 58:\penalty0
  123--152, 2017.

\bibitem[Caragiannis et~al.(2018)Caragiannis, Filos{-}Ratsikas, Nath, and
  Voudouris]{caragiannis2018truthful}
Ioannis Caragiannis, Aris Filos{-}Ratsikas, Swaprava Nath, and Alexandros~A.
  Voudouris.
\newblock Truthful mechanisms for ownership transfer with expert advice.
\newblock \emph{CoRR}, abs/1802.01308, 2018.

\bibitem[Cheng et~al.(2017)Cheng, Dughmi, and Kempe]{cheng2017people}
Yu~Cheng, Shaddin Dughmi, and David Kempe.
\newblock Of the people: Voting is more effective with representative
  candidates.
\newblock In \emph{Proceedings of the 2017 {ACM} Conference on Economics and
  Computation ({EC})}, pages 305--322, 2017.

\bibitem[Cheng et~al.(2018)Cheng, Dughmi, and Kempe]{cheng2018multiple}
Yu~Cheng, Shaddin Dughmi, and David Kempe.
\newblock On the distortion of voting with multiple representative candidates.
\newblock In \emph{Proceedings of the 32nd {AAAI} Conference on Artificial
  Intelligence ({AAAI})}, pages 973--980, 2018.

\bibitem[Fain et~al.(2019)Fain, Goel, Munagala, and Prabhu]{fain2018random}
Brandon Fain, Ashish Goel, Kamesh Munagala, and Nina Prabhu.
\newblock Random dictators with a random referee: Constant sample complexity
  mechanisms for social choice.
\newblock In \emph{Proceedings of the 33rd AAAI Conference on Artificial
  Intelligence ({AAAI})}, 2019.

\bibitem[Feige and Tennenholtz(2010)]{Feige10}
Uriel Feige and Moshe Tennenholtz.
\newblock Responsive lotteries.
\newblock In \emph{Proceedings of the 3rd International Symposium on
  Algorithmic Game Theory ({SAGT})}, pages 150--161, 2010.

\bibitem[Feldman et~al.(2016)Feldman, Fiat, and Golomb]{feldman2016facility}
Michal Feldman, Amos Fiat, and Iddan Golomb.
\newblock On voting and facility location.
\newblock In \emph{Proceedings of the 2016 {ACM} Conference on Economics and
  Computation ({EC})}, pages 269--286, 2016.

\bibitem[Filos-Ratsikas and Miltersen(2014)]{filos2014truthful}
Aris Filos-Ratsikas and Peter~Bro Miltersen.
\newblock Truthful approximations to range voting.
\newblock In \emph{Proceedings of the 10th International Conference on Web and
  Internet Economics ({WINE})}, pages 175--188, 2014.

\bibitem[Filos{-}Ratsikas and Voudouris(2020)]{FRV20}
Aris Filos{-}Ratsikas and Alexandros~A. Voudouris.
\newblock Approximate mechanism design for distributed facility location.
\newblock \emph{CoRR}, abs/2007.06304, 2020.

\bibitem[Filos-Ratsikas et~al.(2014)Filos-Ratsikas, Frederiksen, and
  Zhang]{Aris14}
Aris Filos-Ratsikas, S{\o}ren Kristoffer~Stiil Frederiksen, and Jie Zhang.
\newblock {Social welfare in one-sided matchings: Random priority and beyond}.
\newblock In \emph{Proceedings of the 7th Symposium of Algorithmic Game Theory
  ({SAGT})}, pages 1--12, 2014.

\bibitem[Filos{-}Ratsikas et~al.(2020)Filos{-}Ratsikas, Micha, and
  Voudouris]{FMV19}
Aris Filos{-}Ratsikas, Evi Micha, and Alexandros~A. Voudouris.
\newblock The distortion of distributed voting.
\newblock \emph{Artificial Intelligence}, 286:\penalty0 103343, 2020.

\bibitem[Ghodsi et~al.(2019)Ghodsi, Latifian, and
  Seddighin]{ghodsi2019abstension}
Mohammad Ghodsi, Mohamad Latifian, and Masoud Seddighin.
\newblock On the distortion value of the elections with abstention.
\newblock In \emph{Proceedings of the 33rd AAAI Conference on Artificial
  Intelligence ({AAAI})}, 2019.

\bibitem[Gkatzelis et~al.(2020)Gkatzelis, Halpern, and
  Shah]{gkatzelis2020resolving}
Vasilis Gkatzelis, Daniel Halpern, and Nisarg Shah.
\newblock Resolving the optimal metric distortion conjecture.
\newblock In \emph{Proceedings of the 61st Annual IEEE Symposium on Foundations
  of Computer Science ({FOCS})}, pages 1427--1438, 2020.

\bibitem[Goel et~al.(2017)Goel, Krishnaswamy, and Munagala]{goel2017metric}
Ashish Goel, Anilesh~K Krishnaswamy, and Kamesh Munagala.
\newblock Metric distortion of social choice rules: Lower bounds and fairness
  properties.
\newblock In \emph{Proceedings of the 2017 ACM Conference on Economics and
  Computation (EC)}, pages 287--304, 2017.

\bibitem[Goel et~al.(2018)Goel, Hulett, and Krishnaswamy]{goel2018relating}
Ashish Goel, Reyna Hulett, and Anilesh~K. Krishnaswamy.
\newblock Relating metric distortion and fairness of social choice rules.
\newblock In \emph{Proceedings of the 13th Workshop on Economics of Networks
  ({N}et{E}con)}, page 4:1, 2018.

\bibitem[Goel et~al.(2019)Goel, Krishnaswamy, Sakshuwong, and
  Aitamurto]{goel2016knapsack}
Ashish Goel, Anilesh~K. Krishnaswamy, Sukolsak Sakshuwong, and Tanja Aitamurto.
\newblock Knapsack voting for participatory budgeting.
\newblock \emph{{ACM} Transactions Economics and Computation}, 7\penalty0
  (2):\penalty0 8:1--8:27, 2019.

\bibitem[Gross et~al.(2017)Gross, Anshelevich, and Xia]{gross2017agree}
Stephen Gross, Elliot Anshelevich, and Lirong Xia.
\newblock Vote until two of you agree: Mechanisms with small distortion and
  sample complexity.
\newblock In \emph{Proceedings of the 31st {AAAI} Conference on Artificial
  Intelligence ({AAAI})}, pages 544--550, 2017.

\bibitem[Kempe(2020)]{kempe2019analysis}
David Kempe.
\newblock An analysis framework for metric voting based on lp duality.
\newblock In \emph{Proceedings of the 34th {AAAI} Conference on Artificial
  Intelligence ({AAAI})}, pages 2079--2086, 2020.

\bibitem[Lu and Boutilier(2011)]{lu2011budgeted}
Tyler Lu and Craig Boutilier.
\newblock Budgeted social choice: From consensus to personalized decision
  making.
\newblock In \emph{Proceedings of the 22nd International Joint Conference on
  Artificial Intelligence ({IJCAI})}, pages 280--286, 2011.

\bibitem[Ma et~al.(2020)Ma, Menon, and Larson]{menon2020matching}
Thomas Ma, Vijay Menon, and Kate Larson.
\newblock Improving welfare in one-sided matching using simple threshold
  queries.
\newblock \emph{CoRR}, abs/2011.13977, 2020.

\bibitem[Mandal et~al.(2019)Mandal, Procaccia, Shah, and
  Woodruff]{mandalefficient}
Debmalya Mandal, Ariel~D. Procaccia, Nisarg Shah, and David~P. Woodruff.
\newblock Efficient and thrifty voting by any means necessary.
\newblock In \emph{Proceedings of the 33rd Conference on Neural Information
  Processing Systems (NeurIPS)}, pages 7178--7189, 2019.

\bibitem[Mandal et~al.(2020)Mandal, Shah, and Woodruff]{mandal2020optimal}
Debmalya Mandal, Nisarg Shah, and David~P Woodruff.
\newblock Optimal communication-distortion tradeoff in voting.
\newblock In \emph{Proceedings of the 21st ACM Conference on Economics and
  Computation ({EC})}, pages 795--813, 2020.

\bibitem[Munagala and Wang(2019)]{munagala2019improved}
Kamesh Munagala and Kangning Wang.
\newblock Improved metric distortion for deterministic social choice rules.
\newblock In \emph{Proceedings of the 2019 {ACM} Conference on Economics and
  Computation ({EC})}, pages 245--262, 2019.

\bibitem[Pierczynski and Skowron(2019)]{skowron2019approval}
Grzegorz Pierczynski and Piotr Skowron.
\newblock Approval-based elections and distortion of voting rules.
\newblock In \emph{Proceedings of the 28th International Joint Conference on
  Artificial Intelligence ({IJCAI)}}, pages 543--549, 2019.

\bibitem[Procaccia and Rosenschein(2006)]{procaccia2006distortion}
Ariel~D. Procaccia and Jeffrey~S. Rosenschein.
\newblock The distortion of cardinal preferences in voting.
\newblock In \emph{International Workshop on Cooperative Information Agents
  ({CIA})}, pages 317--331, 2006.

\bibitem[Von~Neumann and Morgenstern(1947)]{vnm}
John Von~Neumann and Oskar Morgenstern.
\newblock Theory of games and economic behavior.
\newblock 1947.

\end{thebibliography}

\end{document}